\documentclass[a4paper,12pt]{amsart}
\usepackage{mathdots, amsmath, amsfonts, amssymb, amsthm, a4wide}
\usepackage{marvosym}
\usepackage{manfnt}
\usepackage{ifsym}
\usepackage{skak}
\usepackage{ccfonts}
\usepackage[T1]{fontenc}

\usepackage{enumerate}
\usepackage{calligra}
\usepackage{graphicx}
\usepackage{color}
\usepackage[arrow]{xy} 	
\usepackage{here}
\usepackage{rotating}
\usepackage[abs]{overpic}
\usepackage{lettrine}
\usepackage{type1cm}
\usepackage{undertilde}
\newtheorem{theorem}{Theorem}[section]
\newtheorem{lemma}[theorem]{Lemma}
\newtheorem{definition}[theorem]{Definition}
\newtheorem{proposition}[theorem]{Proposition}

\newcommand{\fbi}[1]{\textbf{\textit{#1}}}
\newcommand{\brac}{\langle\!\langle}
\newcommand{\ket}{\rangle\!\rangle}

\newcommand{\gf}{\mathfrak{g}}
\newcommand{\hf}{\mathfrak{h}}
\newcommand{\suf}{\mathfrak{su}}

\newcommand{\slf}{\mathfrak{sl}}

\newcommand{\SU}{\mathrm{SU}}

\newcommand{\C}{\mathbb{C}}
\newcommand{\R}{\mathbb{R}}
\newcommand{\yesnumber}{\tag{\theequation}\addtocounter{equation}{1}}
\newcommand{\tcal}[1]{\text{\calligra{#1}}}
\begin{document}
\title[EYM: no magnetic solutions for Abelian models]{There are no magnetically charged particle-like solutions of the Einstein Yang-Mills equations for Abelian models}
\author[M. Fisher]{Mark Fisher}
\address{School of Mathematical Sciences\\
Monash University, VIC 3800\\
Australia}
\email{mark.fisher@monash.edu}

\author[T.A. Oliynyk]{Todd A. Oliynyk}
\address{School of Mathematical Sciences\\
Monash University, VIC 3800\\
Australia}
\email{todd.oliynyk@monash.edu}

%\frontmatter
%\setcounter{tocdepth}{2}	
%\tableofcontents
\begin{abstract}
We prove that there are no magnetically charged particle-like solutions for Abelian models in Einstein Yang-Mills, but for non-Abelian models the possibility remains open. An analysis of the Lie algebraic structure of the Yang-Mills fields is essential to our results. In one key step of our analysis we use invariant polynomials to determine which orbits of the gauge group contain the possible asymptotic Yang-Mills field configurations. Together with a new horizontal/vertical space decomposition of the Yang-Mills fields this enables us to overcome some obstacles and complete a dynamical system existence theorem for asymptotic solutions with nonzero total magnetic charge. We then prove that these solutions cannot be extended globally for Abelian models and begin an investigation of the details for non-Abelian models.
\end{abstract}
\maketitle
\section{Introduction}
Many celebrated results in relativity are proved without having to be too specific about the particular form of the matter content. For example, the positive energy theorem applies for any matter satisfying the dominant energy condition, and global existence for Yang-Mills Higgs is proved for an arbitrary (compact, etc) Yang-Mills gauge group and any arbitrary quartic Higgs potential \cite{EardleyMoncrief1,EardleyMoncrief2}. In contrast to such elegant and powerful results, we will see here that settling the question of whether magnetically charged Einstein Yang-Mills particle-like solutions exist requires us to consider a number of possibilities for the Yang-Mills fields as separate cases.
\par
In particular, for spherically symmetric Einstein Yang-Mills, there are various choices of the Yang-Mills gauge group, $G$, and then different possible spherically symmetric \fbi{models} for the Yang-Mills field. As we will see throughout this article, establishing the existence results requires a detailed description of the Lie-algebraic structure of these fields.  
\par
According to long-standing conjecture, globally regular spherically symmetric (particle-like) solutions with nonzero total magnetic charge are not expected to exist in Einstein Yang-Mills theory \cite{KSUN,BSChernSimons}. The previous evidence for this was based on the observation that in some individual cases the asymptotic behaviour of the proposed solutions ultimately extends to conditions at the origin that are incompatible with the necessary regularity there. We will establish the extent to which this reasoning is valid and in doing so we prove the following:
\begin{theorem} 
For any spherically symmetric Einstein Yang-Mills equations for i) Abelian models, or ii) non-Abelian models arising from a classical Yang-Mills gauge group, asymptotic solutions with non-zero magnetic charge exist for all possible solutions of the non-zero magnetic charge boundary condition. Up to gauge equivalence these solutions are uniquely determined by a finite number of parameters. We also establish the asymptotic fall-off and show that the limit of these solutions is always well-defined.
\end{theorem}
To address the conjecture that these solutions cannot be global, we then prove the following:
\begin{theorem}\label{nomagabelian}
There are no globally regular particle-like solutions with nonzero total magnetic charge for the Einstein Yang-Mills equations in the case of \fbi{Abelian models}.
\end{theorem}
We do not currently see a natural way to generalise this theorem to also include all non-Abelian models, though Theorem \ref{nomagabelian} does generalise to the following cases:
\begin{theorem}\label{nomagnonabelianst}
There are no globally regular particle-like solutions with nonzero total magnetic charge for the Einstein Yang-Mills equations in the case of a \fbi{non-Abelian model} arising from a classical gauge group if proposition \ref{nostrangeeigenvalues} is true (for the non-Abelian model).
\end{theorem}
For the remaining non-Abelian models which do not satisfy this condition, we expect that the details of the asymptotic local existence theorem will serve as a foundation for further investigation of the possibility of magnetically charged particle-like solutions. 
\section{Einstein Yang-Mills}
%Consider a four-dimensional spacetime manifold $M$ equipped with a Lorentzian metric $g$ with components $g_{\mu \nu}$ defined on coordinate charts in the usual way. By choosing a real, compact, semisimple Lie group, $G$, to be the \fbi{Yang-Mills gauge group} we can define a principal bundle $P(M,G)$ and then define a Yang-Mills field as a connection one-form on $P$, which in the same coordinate chart can be written as $A=A_{\mu}dx^{\mu} = A_{\mu}^i T_i dx^{\mu}$, where $\{T_i\}$ is a basis for $\gf_0 = T_e G$, the Lie algebra of $G$. Let $R$ be the Ricci scalar for $g$ and let $F$ be the Yang-Mills curvature two-form i.e. $F = dA + A \wedge A$. Then by varying the action
The Einstein Yang-Mills equations are obtained by minimizing the action
\begin{align}
\int (R- |F|^2)\sqrt{g}d^4x
\end{align}
over all Lorentzian metrics $g$ and Yang-Mills fields $A$, where $R$ is the scalar curvature associated to the metric and $F$ is the Yang-Mills curvature. The resulting equations are 
\begin{align}\label{GAEYM1}\index{Einstein Yang-Mills equations}
R_{\mu \nu}-\frac12 R g_{\mu \nu} &= \brac F_{\mu \alpha} | {F_{\nu}}^{\alpha}\ket -\frac14 \brac F_{\alpha \beta} | F^{\alpha \beta} \ket g_{\mu \nu},\\\label{GAEYM2}
d {\ast}F + &A\wedge {\ast}F -{\ast}F\wedge A = 0,
\end{align}
where $R_{\mu \nu}$ are components of the Ricci tensor, and $\ast$ denotes the Hodge dual associated to the metric $g$. The inner product is positive definite on the Lie algebra and will be defined in a subsequent section. 
\par
The problem of finding particle-like solutions of the Einstein Yang-Mills equations requires that we consider static, spherically symmetric solutions that contain a finite, localised amount of energy and are globally regular and asymptotically flat. In this case the above partial differential equations simplify to a system of ordinary differential equations.
\section{Magnetic Charge}
To define the total (Yang-Mills) magnetic charge, consider $\Sigma$, a spacelike hypersurface of $M$, orthogonal to the time-like Killing vector. $\Sigma$ can be foliated by a family of two spheres, $\mathcal{S}^2_r$, where $r$ is the asymptotic radial coordinate. The total Yang-Mills magnetic charge on $\Sigma$ is then the asymptotic limit of the gauge-invariant flux through the two spheres \cite{AbbottDeser,toddsthesis}, which is given by the formula
\begin{align}\label{QMdef}\index{magnetic charge}\index{$Q_M$}
Q_M = \lim_{r\rightarrow\infty} \frac{1}{4\pi}\int_{\mathcal{S}^2_r} \left\|F_{ab} \epsilon_r^{ab}\right\| \epsilon_r,
\end{align}
where $\epsilon_r = \frac{1}{2}\epsilon_{ab}dx^a\wedge dx^b$ is the area two form on $\mathcal{S}^2_r$ and $F_{ab}$ are the components of the Yang-Mills curvature two-form, $F = \frac{1}{2}F_{ab}dx^a\wedge dx^b$.
\section{Spherical Symmetry}
It is well known how to describe the class of spherically symmetric metrics on a spacetime $M$. We can introduce coordinates $(t,r,\theta,\phi)$ and write the metric as 
\begin{align}\label{ssmetric}\index{spherically symmetric metric}\index{$m(r)$}\index{$S(r)$}
ds^2&=-N(t,r)S(t,r)^2dt^2+(N(t,r))^{-1}dr^2+r^2d\Omega^2,
\end{align}
where $d\Omega^2 = d\theta^2 + \sin^2\theta d\phi^2$ is the metric on the round two-sphere $S^2$, and we define a (quasi-local) mass function, $m(t,r)$, by $N=1-\frac{2m}{r}$. The metric functions $m$ and $S$ are well suited to describing asymptotically flat spacetimes which are regular across $r=0$ and hence also simplify the Einstein equations.
\par
Spherical symmetry for the Yang-Mills fields is more complicated to define because there are many ways to lift an isometry on spacetime to an action on the space of Yang-Mills connections. For real, compact, semisimple gauge groups $G$, it was shown in \cite{B92,BS93} that equivalent spherically symmetric Yang-Mills connections correspond to conjugacy classes of homomorphisms of the isotropy subgroup, $\mathrm{U}(1)$, into $G$. The procedure for choosing a spherically symmetric connection is as follows:
\par 
Let $\gf$ be the complexification of $\gf_0$, the Lie algebra of $G$. Fix a Cartan subalgebra, $\hf$, and obtain the root space decomposition,
\begin{align}\index{Cartan decomposition}
\gf = \hf \oplus \bigoplus_{\alpha \in R^+} \gf_{\alpha}  \oplus \bigoplus_{\alpha \in R^+} \gf_{-\alpha}, 
\end{align}
with root spaces $\gf_{\alpha}$ for $R$ a root system in $\hf^{\ast}$, $R^+$ the subset of positive roots, and $\Delta$ the base of simple roots. The compact real form $\gf_0$ can be recovered from this decomposiiton as 
\begin{align}\label{compactrf}\index{compact real form}
\gf_0 = \bigoplus_{\alpha \in \Delta}i\mathbb{R}\hf_{\alpha} \oplus \bigoplus_{\alpha \in R^+} \mathbb{R}\left(\gf_{\alpha} - \gf_{-\alpha}\right) \oplus \bigoplus_{\alpha \in R^+} \mathbb{R}i\left(\gf_{\alpha} + \gf_{-\alpha}\right) ,
\end{align}
Define the real fundamental Weyl chamber\index{Weyl chamber}, $\mathcal{W}_{\mathbb{R}}$,\index{$\mathcal{W}_{\mathbb{R}}$} as
\begin{align*}
\mathcal{W}_{\mathbb{R}} := \{ H \in \hf_0 | -i\alpha(H) > 0 \,\, \forall \alpha \in \Delta \}
\end{align*}
and denote the closure of $\mathcal{W}_{\mathbb{R}}$ in $\hf_0$ by $\overline{\mathcal{W}}_{\mathbb{R}}$. Define the \fbi{integral lattice}\index{integral lattice}, $\mathcal{I}$\index{$\mathcal{I}$} by
\begin{align*}
\mathcal{I} = \{ X \in \hf_0 | \mathrm{Ad}_{\exp{X}}|_{\gf} = \mathbb{I}_{\gf} \}.
\end{align*}
Choose $\Lambda_0$ such that $2\pi i \Lambda_0 \in \mathcal{I}\cap \overline{\mathcal{W}}_{\mathbb{R}}$.
Define the following subspaces of $\gf$, which are eigenspaces of $\mathrm{ad}_{\Lambda_0}$,
\begin{align}
V_2 &:= \{ x\in \gf \, | \, [\Lambda_0, x ] = 2x\},\\
V_{-2}&:= \{ x\in \gf \, | \, [\Lambda_0, x ] = -2x\},\\
V_{0} &:= \{ x\in \gf \, | \, [\Lambda_0, x ] = 0\}.
\end{align}
Equivalently, $V_2$ is equal to $\bigoplus_{\alpha \in S_2} \gf_{\alpha}$, where $S_n:= \{ \alpha \in R | \alpha(\Lambda_0) = n\}$.\index{$S_n$}
Define a function $\Lambda_+: \mathbb{R} \rightarrow V_2$ and then define $\Lambda_-(r):= -c(\Lambda_+(r))$, where $c$ is the involutive automorphism that defines the compact real form of $\gf$ (e.g. $c(x) = -x^{\dagger}$ for $\gf = \suf(n,\mathbb{C})$). $\Lambda_+$ can be expanded over the root vectors in $V_2$ as 
\begin{align}
\Lambda_+ = \bigoplus_{\alpha \in S_2} w_{\alpha}e_{\alpha},
\end{align}
and then $\Lambda_-$ is similarly expanded over the negative root vectors in $V_{-2}$ as $\overline{w}_{\alpha}e_{-\alpha}$. There is then a natural vector space isomorphism between $V_2$ and $\mathbb{C}^n$ (or $\mathbb{R}^{2n}$), (here $n$ is the dimension of $V_2$). 
\par
With a choice of $\Lambda_0$ fixed the spherically symmetric Yang-Mills connection in the spacetime coordinates above is \cite{B92,BS93}
\begin{align}\nonumber\index{spherically symmetric Yang-Mills connection}
A&= \tilde{A} + \hat{A}\\\label{Aformula}
&=a(t,r)dt+b(t,r)dr+\frac{1}{2}\left(\Lambda_- - \Lambda_+\right)d\theta +\left(\frac{i}{2}\left(\Lambda_- + \Lambda_+\right)\sin{\theta}+\frac{1}{2i}\Lambda_0\cos{\theta}\right)d\phi,
\end{align}
where $a, b,$ the components of $\tilde{A}$, are valued in $\gf_0^{\Lambda_0} := \{\, X\in \gf_0 \, | \, [\Lambda_0,X] = 0 \, \}$.  The $\tilde{A}$ part of $A$ is called the \fbi{Yang-Mills electric} part of $A$, by analogy with the terms in the electromagnetic four-potential. The term $\hat{A}$ is then the \fbi{Yang-Mills magnetic} term. Note that while the individual functions $\Lambda_+, \Lambda_-$ are $\gf$-valued, $A$ is valued in $\gf_0$ overall, as in the decomposition \eqref{compactrf}.

The remaining gauge freedom after spherical symmetry has been imposed is given by the \fbi{residual (gauge) group},\index{residual gauge group} $G_0^{\Lambda_0}$\index{$G_0^{\Lambda_0}$}, defined as the connected Lie group with Lie algebra
\begin{align*}
\gf_0^{\Lambda_0} := \{\, X\in \gf_0 \, | \, [\Lambda_0,X] = 0 \, \},\index{$\gf_0^{\Lambda_0}$}.
\end{align*}
\par
Specifically, we are free to choose a function $g:(t,r) \rightarrow G_0^{\Lambda_0}$ and replace $A$ with $\mathrm{Ad}_{g^{-1}}A + g^{-1}dg$. The resulting transformations for $a,b,\Lambda_+$ are then
\begin{align}
a &\mapsto \mathrm{Ad}_{g^{-1}}a + g^{-1}\frac{\partial g}{\partial t},\\
b &\mapsto \mathrm{Ad}_{g^{-1}}b + g^{-1}\frac{\partial g}{\partial r},\\
\Lambda_+ &\mapsto \mathrm{Ad}_{g^{-1}}\Lambda_+.
\end{align}
It is possible to choose $g$ so that the polar gauge $b(t,r) = 0$ is imposed.
The additional assumptions that the solutions are static and globally regular leads to the conclusion $a(r)=0$. This was shown by Bizo{\'{n}} and Popp for the spherically symmetric $\SU(2)$ equations and the proof generalises for a spherically symmetric model arising from an arbitrary gauge group \cite{bizonpopp}.  Specifically, supposing $a(r) \neq 0$ implies that the solutions can only be black holes.
\section{Abelian and non-Abelian models}\label{abnabmod}
There is an important distinction to be made between two subsets of choices for $\Lambda_0$. When ${2\pi i}\Lambda_0$ is chosen from the interior of the fundamental Weyl chamber, we have
\begin{align*}
\gf_0^{\Lambda_0} = \hf_0
\end{align*}
and the residual gauge group will be completely \fbi{Abelian} and we call such models \fbi{Abelian models}\index{Abelian models}.
\par
However, for choices of ${2\pi i}\Lambda_0$ which lie on the wall of the Weyl chamber there will be a set of roots $\alpha$ such that $[\Lambda_0,e_{\pm\alpha}] = 0 $ and then the residual group will be \fbi{non-Abelian}\index{non-Abelian models}\footnote{In some cases the non-Abelian residual group will have an Abelian action on $V_2$, suggesting the possibility of a stricter classification than the one we employ.}.
To illustrate the distinction between the two kinds of spherically symmetric models, consider the following two possible spherically symmetric models in $\SU(4)$:

\begin{align}\label{su4abelian}
\Lambda_0 = \left[\begin{array}{cccc} 3 & 0 & 0 & 0 \\ 0 & 1 & 0 & 0 \\ 0 & 0 & -1 & 0 \\ 0 & 0 & 0 & -3\end{array}\right]\!, & 
\,\,\Lambda_+\! = \left[\begin{array}{cccc} 0 & w_1 & 0 & 0 \\ 0 & 0 & w_2 & 0 \\ 0 & 0 & 0 & w_3 \\ 0 & 0 & 0 & 0\end{array}\right]\!, &
\,\,\Lambda_-\! = \left[\begin{array}{cccc} 0 & 0 & 0 & 0 \\ \overline{w}_1 & 0 & 0 & 0 \\ 0 & \overline{w}_2 & 0 & 0 \\ 0 & 0 & \overline{w}_3 & 0\end{array}\right];
\end{align}

\begin{align}\label{su4nonabelian}
\Lambda_0 = \left[\begin{array}{cccc} 2 & 0 & 0 & 0 \\ 0 & 0 & 0 & 0 \\ 0 & 0 & 0 & 0 \\ 0 & 0 & 0 & -2\end{array}\right]\!, & 
\,\,\Lambda_+\! = \left[\begin{array}{cccc} 0 & w_1 & w_2 & 0 \\ 0 & 0 & 0 & w_3 \\ 0 & 0 & 0 & w_4 \\ 0 & 0 & 0 & 0\end{array}\right]\!, &
\,\,\Lambda_-\! = \left[\begin{array}{cccc} 0 & 0 & 0 & 0 \\ \overline{w}_1 & 0 & 0 & 0 \\ \overline{w}_2 & 0 & 0 & 0 \\ 0 & \overline{w}_3 & \overline{w}_4 & 0\end{array}\right].
\end{align}
For the first model, the elements of $\suf(4)$ which commute with $\Lambda_0$ are contained in the Cartan subalgebra and are of the form
\begin{align}
\left[\begin{array}{cccc} ic_1 & 0 & 0 & 0 \\ 0 & -ic_1+ic_2 & 0 & 0 \\ 0 & 0 & -ic_2 +ic_3 & 0 \\ 0 & 0 & 0 & -ic_3\end{array}\right].
\end{align}
These generate the Abelian residual group $G_0^{\Lambda_0} = U(1)^3$. For the second model the elements of $\suf(4)$ which commute with $\Lambda_0$ are of the form

\begin{align}
\left[\begin{array}{cccc} ic_1 & 0 & 0 & 0 \\ 0 & -ic_1+ic_2 & c_4 + ic_5 & 0 \\ 0 & -c_4+ic_5 & -ic_2 +ic_3 & 0 \\ 0 & 0 & 0 & -ic_3\end{array}\right].
\end{align}
Such elements generate the non-Abelian residual group $G_0^{\Lambda_0} = \SU(2)\times U(1)^2$.
\section{The Spherically Symmetric Einstein Yang-Mills Equations}
Using equation \eqref{Aformula} to calculate the curvature, we have 
\begin{align*}
F =& \frac{1}{2}(\Lambda_+' - \Lambda_-')dr\wedge d\theta + \frac{i}{2}(\Lambda_+' + \Lambda_-')\sin\theta dr\wedge d\phi\\
 &+ \frac{1}{2}\left\|\Lambda_0-[\Lambda_+,\Lambda_-]\right\|\sin\theta d\theta\wedge d\phi.
\end{align*}
By calculating the Hodge dual, we also have 
\begin{align*}
\ast F =& \frac{1}{2}(\Lambda_+' - \Lambda_-')SN\sin\theta dt \wedge d\phi - \frac{i}{2}(\Lambda_+' + \Lambda_-')SN dt\wedge d\theta\\
 &+ \frac{1}{2}\left\|\Lambda_0-[\Lambda_+,\Lambda_-]\right\|\frac{S}{r^2} dt\wedge dr.
\end{align*}
Then, by substituting the above expressions into the Einstein Yang-Mills equations, \eqref{GAEYM1} and \eqref{GAEYM2}, we obtain, in terms of $m,S,\Lambda_{\pm}$, the \fbi{purely magnetic, static, spherically symmetric, Einstein Yang-Mills equations}\cite{B92,BS93}:
\begin{align*}
m'&=\frac{N}{2}\left\|\Lambda_+'\right\|^2+\frac{1}{8r^2}\left\|\Lambda_0-[\Lambda_+,\Lambda_-]\right\|^2,&\\
S^{-1}S'&=\frac{\left\|\Lambda_+'\right\|^2}{r},&\\
0&=[\Lambda_+,\Lambda_-']+[\Lambda_-,\Lambda_+'],&\\\yesnumber
0&=(NS\Lambda_+')'+\frac{S}{r^2}\left(\Lambda_+-\frac{1}{2}\left[\left[\Lambda_+,\Lambda_-\right],\Lambda_+\right]\right).&
\end{align*}
The variable $S$ can be decoupled from the system at the expense of an extra term in the second order equations. In this case the equations become
\begin{align}
m' =\frac{N}{2}||\Lambda_+'||^2+&\frac{1}{8r^2}||\Lambda_0-[\Lambda_+,\Lambda_-]||^2, \label{meqn}\\
r^2N\Lambda_+''+2(m-\frac{1}{8r}||\Lambda_0-[\Lambda_+,\Lambda_-]||^2)&\Lambda_+'+\Lambda_+-\frac{1}{2}[[\Lambda_+,\Lambda_-],\Lambda_+]=0, \label{ym2} \\
[\Lambda_+',\Lambda_-]+&[\Lambda_-',\Lambda_+]=0, \label{ym1} \\
S^{-1}S=&\frac{1}{r}||\Lambda_+'||^2. \label{Seqn}
\end{align}
\section{The Inner product}
The $||\cdot||$-norm is associated to the inner product $\brac \cdot | \cdot \ket$, defined by
\begin{align}
\brac X| Y \ket = k\mathrm{Re} \left( -c(X) | Y \right),
\end{align}
where $k$ is a positive number we are free to scale, $( \cdot | \cdot)$ is the Killing form on $\gf$, and $c$ is an \fbi{involutive automorphism} on $\gf$ defining the compact real form, i.e. 
\begin{align*}\index{involutive automorphism}
c(x + iy) = x - iy \quad \forall x,y \in \gf_0,
\end{align*}

(hence the involutive property, $c^2 = 1$). It follows from the properties of the Killing form that our inner product is positive definite on $\gf$, and satisfies
\begin{align}\label{innerproductproperties}
\brac X | Y \ket &= \brac Y | X \ket,\\
\brac c(X) | c(Y) \ket &= \brac X | Y \ket,\\
\brac [X,c(Y) | Z] \ket &= \brac X | [Y,Z] \ket.
\end{align}

There is freedom to rescale the $||\cdot||$-norm by a constant factor which leads to a global rescaling of $m$ and $r$:
\begin{proposition}\label{propscale}
	If $(m(r),\Lambda_+(r))$ satisfies \eqref{meqn}-\eqref{ym1}, then 
$$(\alpha
        m(r/\alpha),\Lambda_+(r/\alpha))$$ satisfies the equations obtained by replacing $||\cdot||$ in
        equations \eqref{meqn}-\eqref{ym1} with $\alpha||\cdot||$. 
	\end{proposition}
	\begin{proof}
	Substitution.
	\end{proof} 
There is a family of $SU(2)$ (Bartnik-McKinnon) solutions known to be embedded in all regular models and we conventionally scale the inner product for a given model so the numerical parameters for this family match those given in \cite{BFM,KKSU3,KSUN}. For the equations as written here, this scaling corresponds to fixing $\left\|\Lambda_0\right\|^2 = 4$.
\par
The above properties of the inner product are very useful for proving geometrical statements about $V_2$. For example we have the following proposition
\begin{proposition} If $\Lambda_+\in V_2$ then
\begin{align}
\frac{\left\|[\Lambda_+,\Lambda_-]\right\|^2}{\left\|\Lambda_+\right\|^4} \geq \frac{4}{\left\|\Lambda_0\right\|^2}.
\end{align}
\end{proposition}
\begin{proof}
By the properties of the inner product we have that
\begin{align*}
&\brac [\Lambda_+,\Lambda_-] | \Lambda_0 \ket\\
&=\brac \Lambda_+ | [\Lambda_0, \Lambda_+] \ket\\
&=\brac \Lambda_+ | 2\Lambda_+ \ket\\
&=2\left\|\Lambda_+\right\|^2.
\end{align*}
Considering this term as the projection of the commutator onto $\Lambda_0$ we then have that
\begin{align*}
[\Lambda_+,\Lambda_-] = \frac{2\left\|\Lambda_+\right\|^2}{\left\|\Lambda_0\right\|^2}\Lambda_0 + \text{orthogonal terms}
\end{align*}
and therefore (by the `Hypotenuse Inequality'), it follows that
\begin{align}
\frac{\left\|[\Lambda_+,\Lambda_-]\right\|^2}{\left\|\Lambda_+\right\|^4} \geq \frac{4}{\left\|\Lambda_0\right\|^2}.
\end{align}
\end{proof}\index{coercive condition}
This proposition implies that Oliynyk and K\"unzle's `coercive condition' \cite{OK03} always holds, and will also be useful to us in the proof that Abelian models do not possess magnetically charged solutions.
\section{Asymptotic Behaviour and Magnetic Charge}
Requiring that the solutions are regular and asymptotically flat gives boundary conditions\footnote{The condition at the origin follows immediately from the requirement that the equations are not singular at the origin. The asymptotic condition follows from demonstrating that the system is asymptotically autonomous and determining the critical points of the autonomous dynamical system, see \cite{OK03}.}
\begin{align}\label{requalszerocondition}
[\Lambda_+,\Lambda_-] &= \Lambda_0, &\text{at $r=0$,}\\\label{infinitycondition}
[[\Lambda_+,\Lambda_-],\Lambda_+] &= 2\Lambda_+, &\text{as $r \rightarrow \infty$.}
\end{align}%
at $r=0$, and as $r\rightarrow\infty$, respectively.
In \cite{OK03}, Oliynyk and K\"unzle proved that any bounded solution to the static spherically symmetric Einstein Yang-Mills equations on $[r_0,\infty)$ (satisfying appropriate bounds at $r_0 > 0$) will satisfy the properties
\begin{align}\label{OKasymptote}
r\Lambda_+'(r)\rightarrow 0, \quad \left\|\Lambda_+(r) - \mathfrak{F}^{\times}\right\| \rightarrow 0 \quad \text{as} \quad r \rightarrow \infty,
\end{align}
where $\mathfrak{F}^{\times} := \left\{ X\in V_2\backslash\left\{0\right\} \,\, | \,\, [[c(X),X],X] = 2X \right\}$ and the distance is the infimum of the norm over $\mathfrak{F}^{\times}$.
\par
For the purely magnetic, static, spherically symmetric Einstein Yang-Mills equations (when written in the polar gauge), the Yang-Mills curvature $F$ is given by the formula
\begin{align}
F = \frac{1}{2}(\Lambda_+' - \Lambda_-')dr\wedge d\theta + \frac{i}{2}(\Lambda_+' \!-\! \Lambda_-')dr\wedge d\phi + \frac{i}{2}(\Lambda_0-[\Lambda_+,\Lambda_-])\sin\theta d\theta \wedge d\phi.
\end{align}
It then follows from \eqref{QMdef} that the expression for the total magnetic charge simplifies to 
\begin{align}
4Q_M = \lim_{r\rightarrow\infty}\left\|\Lambda_0-[\Lambda_+,\Lambda_-]\right\|.
\end{align}
If $\Lambda_+$ has a limit $\Lambda_+^{\infty}$ as $r\rightarrow\infty$ and the limit satisfies $[\Lambda_+^{\infty},\Lambda_-^{\infty}]=\Lambda_0$ (which is also the boundary condition at the origin), then the total magnetic charge will be zero. If $\Lambda_+^{\infty}$ does not satisfy this equation then the total magnetic charge will be nonzero.
\par
Since the Adjoint action of $G_0^{\Lambda_0}$ takes solutions of the equations to equivalent solutions, we are interested in the orbit space $ \mathfrak{F}^{\times}\slash G_0^{\Lambda_0}$. Since $G_0^{\Lambda_0}$ is a compact group, the structure of the orbit space can be understood by considering the invariant polynomials.
\section{The Invariant Polynomials $\mathbb{R}\left[\left(V_2\oplus V_{-2}\right)_0\right]^{G_0^{\Lambda_0}}$} 
\begin{theorem}[Hilbert Weyl Theorem]
The ring of invariant polynomials for a representation of a compact Lie group acting on a real vector space is finitely generated. 
\end{theorem}
We now establish some additional results for the ring  $\mathbb{R}\left[\left(V_2\oplus V_{-2}\right)_0\right]^{G_0^{\Lambda_0}}$.The notation $\left(V_2\oplus V_{-2}\right)_0$ denotes the space where the coefficient of $e_{-\alpha}$ is set as the complex conjugate of the coefficient of $e_{\alpha}$ and when the subscript $0$ is omitted no such relations have been imposed (the complexification).
\begin{theorem}\label{AbelINV}
For any spherically symmetric Abelian model, the ring of invariant polynomials,
\begin{align}
\R\left[\left(V_2\oplus V_{-2}\right)_0\right]^{G_0^{\Lambda_0}},
\end{align}
is generated by the set of quadratic polynomials $\{|w_{\alpha}|^2,\,\,\,\alpha\in S_2\}$.
\end{theorem}
\begin{proof}
Since the model is Abelian, $G_0^{\Lambda_0} = H_0$, the Abelian group generated by the Cartan subalgebra (i.e. the maximal torus). The action on the root spaces is therefore diagonal. Furthermore, since the root spaces are one-dimensional, for any two distinct root spaces we can always find an element of the Cartan subalgebra with an adjoint action that multiplies each root space by a different factor. The corresponding action of $H_0$ will be to multiply the root vectors by two different phases. It follows that a polynomial can only be invariant if it pairs the coefficient of $e_{\alpha}$, ($w_{\alpha}$), with the coefficient of $e_{-\alpha}$ ( $\overline{w}_{\alpha}$). It then follows that all the polynomials are even order and every term is a product of quadratics of the form $|w_{\alpha}|^2$. Hence the set $\{ |w_{\alpha}|^2, \alpha \in S_2 \}$ generates $\R\left[\left(V_2\oplus V_{-2}\right)_0\right]^{G_0^{\Lambda_0}}$.
\end{proof}
For the non-Abelian case we have the following theorem:
\begin{theorem}\label{nonAbelINV}
For any spherically symmetric model arising from classical group, (type $A,B,C,D$), the ring of invariant polynomials,
\begin{align}
\R\left[\left(V_2\oplus V_{-2}\right)_0\right]^{G_0^{\Lambda_0}},
\end{align}
is generated by polynomials that are either the real or the imaginary part of the trace or the polarized Pfaffian of a product of the block matrices in $\left(V_2\oplus V_{-2}\right)_0$ that form irreducible representations of $G_0^{\Lambda_0}$.
\end{theorem}
Remark: The polarized Pfaffian is only needed for groups from the D series.
\par
\begin{proof}
This theorem follows as an application of results of Aslaksen, Tan, and Zhu \cite{ATZ1}. There they consider the invariant polynomials for the action of the group $L$ generated by the Levi factor of a parabolic subalgebra of a classical Lie algebra $G$, acting on the Lie algebra $\gf$. We can recognise the relevance of this scenario to our problem by considering the Jacobson-Morozov parabolic subalgebra associated to $\Lambda_0$ \cite{CollingwoodandMcGovern}.
\par
In \cite{ATZ1} they decompose the fundamental representation of $\gf$ into irreducible representations of the Levi factor and then show that the blocks of the matrices and the corresponding subspaces of the vector space form a representation of a suitable quiver.
A quiver is a directed graph and a representation of a quiver associates vector spaces to the vertices and maps between vector spaces to the edges.
\par
The graph properties of the quiver facilitate a way of turning the invariant theory problem on the blocks of $\gf$ (the edges of the quiver graph), into an equivalent problem on the subspaces of the representation vector (the vertices of the quiver graph). This problem can then be solved using the first fundamental theorems from Classical Invariant theory and the result translated into a solution of the original problem.
\par
The graph properties of the quiver also refine the description of which products of block matrices are sufficient, since a possible product considered as the representation of a concatenation of directed edges must be a closed path in the quiver graph. The path is only allowed to pass through a vertex associated to an $n$-dimensional vector space at most $n^2$ times.
\par
The results in \cite{ATZ1} are for classical groups of type $A,B,$ or $C$, but for our application we were able to include the $D$ case by making use of the related results of the same authors in \cite{ATZ2}. More details of how to apply their results to the various cases can be found in \cite{MFthesis}.
\end{proof}
 
Example: Consider again the examples of Abelian and non-Abelian models in equations \eqref{su4abelian} and \eqref{su4nonabelian}. For the Abelian model in \eqref{su4abelian}, we have from Theorem \ref{AbelINV} that
$$\mathbb{R}\left[\left(V_2\oplus V_{-2}\right)_0\right]^{G_0^{\Lambda_0}} = \mathbb{R}\left[|w_1|^2, |w_2|^2, |w_3^2\right]$$,
whereas for the non-Abelian model in equation \eqref{su4nonabelian}, we must consider possible products. It is possible to determine the specific orders of the generators by calculating the Molien function, where the formulas obtained by Forger in \cite{Forger} are particularly useful for our application. The Molien function for this model is 
$$\frac{1}{(1-|z|^2)^2(1-|z|^4)}$$
from which one can read off that there are two quadratic generators and one quartic generator. By letting $\Lambda_+^1,\Lambda_+^2$ be the block matrices
\begin{align}
\Lambda_+^1 = \left[\begin{array}{cc}w_1 & w_2\end{array}\right],\quad \Lambda_+^2 = \left[\begin{array}{c}w_3 \\ w_4\end{array}\right],
\end{align}
we can associate this model to the quiver representation in Figure \ref{su4naquiv}.
\begin{figure}[H]
	\begin{center}
	\begin{overpic}[scale=0.8]{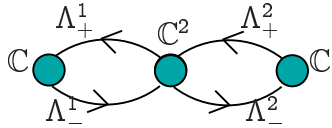}
	\put(9,33){$\Lambda_+^1$}
	\put(78,33){$\Lambda_+^2$}
	\put(5,-2){$\Lambda_-^1$}
	\put(80,-2){$\Lambda_-^2$}
	\put(-9,17){$\C$}
	\put(47,27){$\C^2$}
	\put(104,17){$\C$}
	\end{overpic}
	\caption{The quiver for a non-Abelian model in $SU(4)$}
	\label{su4naquiv}
	\end{center}
\end{figure}

It follows from the proof of Theorem \ref{nonAbelINV} the allowed products of the block matrices must form closed paths in the quiver  that only visit an $n$-dimensional vertex at most $n^2$ times. We can use this to identify the following generators for the ring of invariant polynomials:
\begin{align*}
\mathbb{R}\left[\left(V_2\oplus V_{-2}\right)_0\right]^{G_0^{\Lambda_0}} &= \mathbb{R}\left[\mathrm{Tr}(\Lambda_+^1\Lambda_-^1),\mathrm{Tr}(\Lambda_+^2\Lambda_-^2), \mathrm{Tr}(\Lambda_-^2\Lambda_-^1\Lambda_+^1\Lambda_+^2)\right]\\
&= \mathbb{R}\left[|w_1|^2+|w_2|^2, |w_3|^2+|w_4|^2, |w_1w_3 + w_2w_4|^2 \right].
\end{align*}
We have shown in \cite{MFthesis, BFO1} that it is possible to project the system of ordinary differential equations onto the space of invariant polynomials and consequently reduce the number of equations that need to be solved to the smallest number, while still retaining all the gauge-invariant quantities. This can be quite a simplification, for example in the non-Abelian $SU(4)$ model above, the eight real second-order differential equations for each of the $w_i$ can be reduced to three second-order equations for the invariant polynomials.  In \cite{BFO1}, we obtained a similar simplification for a non-Abelian model in $SO(5)$, where six equations were reduced to two and we were consequently able to overcome the previously intractable numerical problem of finding solutions.
\par
In the present context our main interest in the invariant polynomials is in using them to determine information about the orbits making up the asymptotic condition. To that end we prove the following:
\begin{theorem}\label{commdeterminesthm}
For either any spherically symmetric Abelian model, or any spherically symmetric model arising from a classical group, an equation of the form
\begin{align}\label{comPhi}
\left[\Lambda_+, \Lambda_-\right] = \Phi,
\end{align}
with prescribed right hand side $\Phi$, will \fbi{at least} determine the value of all invariant polynomials.
\end{theorem}
Remark: By `at least' we mean that the further possibility of an overdetermined system of equations is not ruled out. For example the boundary condition at the origin, $[\Lambda_+,\Lambda_-] = \Lambda_0$ is of the above form, and examples are known of spherically symmetric models that have no solutions to this equation.
\begin{proof}
In the Abelian case, the invariants are all quadratic and this theorem is proved in the same way as Lemma 8.2 in \cite{OK03}.
\par
Otherwise, obtain the triangular decomposition of $\gf$, so that $V_2$ is made up of block matrices that are in the strictly upper triangular blocks of $\gf$, $V_2$ is made up of the corresponding strictly lower triangular blocks, and $\gf_0^{\Lambda_0} = V_0$ is made up of diagonal blocks.
\par
The blocks in $V_2$ can then be ordered as follows: Identify the first block row containing a block in $V_2$, and label this block $\Lambda_+^1$. There are now two possibilities to consider:
\begin{enumerate}[(a)]
\item $\gf$ or $\Lambda_+^1$ do not meet the conditions in item (b),
\item $\gf$ is an orthogonal Lie algebra ($B$ or $D$ type) and $\Lambda_+^1$ is a block with only one column.
\end{enumerate}
For case (a) there cannot be a second block in $V_2$ in the same block row as $\Lambda_+^i$ since, if there were, then $[V_2,V_{-2}]=V_0$ would have entries that are not in the block diagonal. This argument does not work for case (b) because there are some zeroes down the antidiagonal in the triangular representation of orthogonal matrices. There can therefore be a second block with one column in the same block row as $\Lambda_+^1$ with no contradiction to the commutator condition.
\par
For blocks that satisfy case (a) we can form sequences as follows: consider the block row and block column that $\Lambda_+^i$ is in, say $(j,k)$. If block row $k$ contains a block that is in $V_2$ then this block is next in the sequence and if it is empty this sequence terminates and the next sequence begins with the block in $V_2$ that is in the first block row of $\gf$ to not yet be considered. Eventually all blocks of $V_2$ will be in either a sequence $\Lambda_+^1,\ldots,\Lambda_+^{n_1},\Lambda_+^{n_1+1},\ldots,\Lambda_+^{n_2}, \text{etc.}$, or in an \fbi{orthogonal exception} of the form
\par 
\begin{align}
\left[\begin{array}{ccccc}
0 & \Lambda_+^a & 0 & \Lambda_+^b & 0 \\
0 & 0 & 0 & 0 & -(\Lambda_+^b)^t\\
0 & 0 & 0 & 0 & 0\\
0 & 0 & 0 & 0 & (-\Lambda_+^a)^t\\
0 & 0 & 0 & 0 & 0
\end{array}\right],
\end{align}
where $\Lambda_+^a,\Lambda_+^b$ are column one blocks (note that the orthogonal transpose is not the standard operation in the triangular decomposition). This exhausts all possibilities for the blocks in $V_2$ and the description for $V_{-2}$ is the same with transposed indices.
\par
The theorem then follows in three steps as a proof by induction: 1) show that the invariants are all even order, 2) show that \eqref{comPhi} determines the quadratic invariants (base step), 3) show that \eqref{comPhi} can be used to give the value of higher order invariants in terms of lower order invariants and known quantities (induction step).
\par
For the first step we simply consider the Adjoint action of $\mathrm{exp}(i\theta\Lambda_0)$ on $V_2\oplus V_{-2}$, which multiplies the negative root vectors by $e^{-2i\theta}$ and the positive root vectors by $e^{2i\theta}$. Hence any invariant polynomial must consist of products of an equal number of of positive root vector coefficients and negative root vector coefficients, implying only even polynomials.
\par
Step 2 is a generalisation of Lemma 8.2 in \cite{OK03}, where we use the block row and column description of $V_2$ to introduce a suitable basis for $\hf$. For $A$-type $\gf$, this is constructed by taking each block row in $\gf$ containing a block $\Lambda_+^i$ in $V_2$ and then introducing a basis vector for $h$ of the form $\mathrm{diag}(0,a,\ldots,a,-b\ldots,-b)$, where the number of zeroes is equal to the number of rows before the block row containing $\Lambda_+^i$, $a$ is a positive integer which appears as many times as there are rows in $\Lambda_+^i$ and is equal to the number of times $b$ appears. $b$ is a positive integer that fills the remaining diagonal entries and is equal to the number of times $a$ appears, (hence making the basis vector traceless). The adjoint action of such an element of $\hf$ will be negative on blocks in rows above $\Lambda_+^i$, positive on $\Lambda_+^i$, and zero on subsequent blocks. A set of such basis vectors for $\hf$ will then provide the appropriate generalisation of Lemma 8.2 in \cite{OK03}, since the resulting system of equations for the quadratic invariant polynomials will be in echelon form and hence invertible by elementary methods. The basis vectors for type $B,C,D$ are similar with the appropriate reflection symmetry.
\par
For step 3, we can use the ordering on the blocks of $V_2$ to expand \eqref{comPhi} as the system of equations
\begin{align*}
\Lambda_+^1\Lambda_-^1 = \Phi_{1},\\
\ldots,\\
\Lambda_+^j\Lambda_-^j - \Lambda_-^{j-1}\Lambda_+^{j-1} = \Phi_{j},\\
\ldots,\\
-\Lambda_-^{n_1}\Lambda_+^{n_1}=\Phi_{n_1},\\
\text{etc.},
\end{align*}
for the sequences of blocks, and 
\begin{align*}
\Lambda_+^a\Lambda_-^a + \Lambda_+^{b}\Lambda_-^{b}=\Phi_a,\\
-\Lambda_-^a\Lambda_+^a + (\Lambda_+^{b})^t(\Lambda_-^{b})^t=\Phi_b,
\end{align*}
for the orthogonal exceptions, where the various $\Phi_i$ are known matrix quantities obtained from projecting $\Phi$ on the various block diagonal matrices.
\par
Because we have the description of the invariant polynomials from Theorem \ref{nonAbelINV} as products that form paths in the appropriate quiver, these systems of quadratic matrix equations for the sequences can be systematically used to turn the trace (or polarized Pfaffian) of a product of $p$ block matrices into a combination of lower order invariant polynomials, and a trace (or polarized Pfaffian) of $(\Lambda_+^{n_k+1}\Lambda_-^{n_k+1})^{\frac{p}{2}}$ which are all determined from the first equation of each sequence and the induction hypothesis. For the orthogonal exceptions, the finiteness condition in Theorem \ref{nonAbelINV} and the fact that $\Lambda_+^a,\Lambda_+^b$ have only one column implies that we only need to consider at most quartic products of the exceptional blocks which can easily be enumerated and shown to be determined by the system.
\end{proof}
Remark: More details can be found in \cite{MFthesis}.
Example: As a simple example of the above theorem consider again the non-Abelian model in \eqref{su4nonabelian}. In this case the equation \eqref{comPhi} becomes
\begin{align}
\Lambda_+^1\Lambda_-^1 = \Phi_1,\\
\Lambda_+^2\Lambda_-^2 - \Lambda_-^1\Lambda_+^1 = \Phi_2,\\
-\Lambda_-^2\Lambda_+^2 = -\Phi_3,
\end{align}
where $\Phi_1,\Phi_2$, and $\Phi_3$ are square matrices of sizes $1,2$, and $1$ respectively. The invariants can then be determined by
\begin{align}
\mathrm{Tr}(\Lambda_+^1\Lambda_-^1) &= \mathrm{Tr}(\Phi_1),\\
\mathrm{Tr}(\Lambda_+^2\Lambda_-^2) &= \mathrm{Tr}(\Phi_3),\\
\mathrm{Tr}(\Lambda_-^2\Lambda_-^1\Lambda_+^1\Lambda_+^2) 
&= \mathrm{Tr}(\Lambda_-^1\Lambda_+^1\Lambda_+^2\Lambda_-^2)\\
&= \mathrm{Tr}(\Lambda_-^1\Lambda_+^1(\Lambda_-^1\Lambda_+^1+\Phi_2))\\
&= \mathrm{Tr}((\Lambda_-^1\Lambda_+^1)^2)+\mathrm{Tr}(\Lambda_-^1\Lambda_+^1\Phi_2)\\
&= \mathrm{Tr}((\Phi_1)^2)+\mathrm{Tr}(\Lambda_-^1\Lambda_+^1\Phi_2),
\end{align}
where since the last term is quadratic it is determined by some combination of the first two equations.
\par
Equipped with the above results, we will now return to the investigation of the asymptotic behaviour.
\section{Asymptotic Behaviour and Magnetic Charge - II}
Recall from \S{8} that the boundary conditions for globally regular particle-like solutions are
\begin{align}\label{requalszerocondition2}
[\Lambda_+,\Lambda_-] &= \Lambda_0, &\text{at $r=0$,}\\\label{infinitycondition2}
[[\Lambda_+,\Lambda_-],\Lambda_+] &= 2\Lambda_+, &\text{as $r \rightarrow \infty$.}
\end{align}
Let us first consider an equation of the form
\begin{align}\label{commequalsA}
[\Lambda_+,\Lambda_-] &= \Phi,
\end{align}
for some prescribed $\Phi$. 
\begin{lemma} If $\tilde{\Lambda}_+ \in V_2$, is conjugate to $\Lambda_+$ (under the Adjoint action of  $G_0^{\Lambda_0}$), i.e. $\tilde{\Lambda}_+=g\cdot\Lambda_+$, and $\Lambda_+$ satisfies \eqref{commequalsA}, then $\tilde{\Lambda}_+$ satisfies
\begin{align}
[\tilde{\Lambda}_+,\tilde{\Lambda}_-] = g\cdot\Phi.
\end{align}
\end{lemma}
\begin{proof}
\begin{align*}
[\tilde{\Lambda}_+,\tilde{\Lambda}_-] &= [g\cdot\Lambda_+,g\cdot\Lambda_-]\\&= g\cdot[\Lambda_+,\Lambda_-] \\&= g\cdot \Phi.
\end{align*}
\end{proof}
\begin{proposition}
Solutions of \eqref{requalszerocondition2} all lie on the same orbit and all points on that orbit are solutions.
\end{proposition}
\begin{proof}
Eq. \ref{requalszerocondition} is a case of equation \eqref{commequalsA}, where the prescribed right hand side is $\Phi=\Lambda_0$. Since $g\cdot\Lambda_0=\Lambda_0$ (by definition of $G_0^{\Lambda_0}$), any conjugation of any one solution to \eqref{requalszerocondition} is again a solution to \eqref{requalszerocondition}. To prove that there is only one orbit we invoke Theorem \ref{commdeterminesthm} and the fact that invariant polynomials separate orbits.
\end{proof}
Remark: This result is established a different way in \cite{OK03}.
\par
Next consider equation \eqref{infinitycondition2}.
If $\Lambda_+$ has a well defined limit as $r\rightarrow \infty$, say  $\Lambda_+ \rightarrow \Omega_+$, then $\Omega_+$ will satisfy \eqref{infinitycondition2}, 
\begin{align}
[[\Omega_+,\Omega_-],\Omega_+]=2\Omega_+.
\end{align}
\begin{proposition}
Any point in $V_2$ on the same orbit as $\Omega_+$, e.g. $\tilde{\Omega}_+=g\cdot\Omega_+$ is also a solution to \eqref{infinitycondition2}.
\end{proposition}
\begin{proof}
\begin{align*}
[[\tilde{\Omega}_+,\tilde{\Omega}_-],\tilde{\Omega}_+]&=[[g\cdot\Omega_+,g\cdot\Omega_-],g\cdot\Omega_+]\\
&=2g\cdot\Omega_+\\
&=2\tilde{\Omega}_+.
\end{align*}
\end{proof}
Hence any point on the same orbit as a solution to \eqref{infinitycondition2} is a (gauge-equivalent) solution to \eqref{infinitycondition2} as well. We can use this fact to choose a representative element from each orbit of points that solve \eqref{infinitycondition2}:
\begin{proposition}
If $\tilde{\Omega}_+$ solves \eqref{infinitycondition2} then there exists $\Omega_+ = g\cdot\tilde{\Omega}_+$ that also solves \eqref{infinitycondition2} and has the property:
\begin{align}
[\Omega_+,\Omega_-] \in i\mathfrak{h}_0.
\end{align}
\end{proposition}
\begin{proof}
From the facts that $\Omega_{\pm}\in V_{\pm 2}$ and $c(\Omega_+)=-\Omega_-$ it follows that $[\Omega_+,\Omega_-] \in i\mathfrak{g}_0^{\Lambda_0}$. It can be shown (by a well-known application of the Lefschetz fixed point theorem \cite{Knapp}) that any element in the Lie algebra, $\mathfrak{g}_0$, of a compact Lie group, $G$, can be $\mathrm{Ad}_G$-conjugated to $\mathfrak{h}_0$. In this case $\mathfrak{g}=\mathfrak{g}_0^{\Lambda_0}$ and $G=G^{\Lambda_0}_0$ and in light of the previous proposition the result follows.
\end{proof}
\par
For each orbit of solutions of \eqref{infinitycondition2}, let $\Omega_+^i$ be a representative element as defined in the above proposition. The index $i$ is over the set of orbits of solutions and we will subsequently show this index to be discrete. 
\par
Define $\Omega_0^i:=[\Omega_+^i,\Omega_-^i] \in i\mathfrak{h}_0$ and define the spaces:
\begin{align}
V_{\pm 2}^{\Omega_0^i}:=\left\{X\in\mathfrak{g}|[\Omega_0^i,X]=\pm2X\right\}.
\end{align}\index{$V_{\pm 2}^{\Omega_0^i}$}
Remark: Hence $V_2$ could also be written as $V_2^{\Lambda_0}$.
\begin{proposition}
The set of possible $\Omega_0^i$ is discrete. 
\end{proposition}
\begin{proof}
Each $\Omega_0^i$ was defined in terms of (at least one) $\Omega_+^i$. Hence there is a $\slf_2\mathbb{C}$ subalgebra of $\gf$ defined by $\mathrm{span}_{\mathbb{C}}[\Omega_0^i,\Omega_+^i,\Omega_-^i]$. Therefore $\gf$ can be decomposed as an $\slf_2\mathbb{C}$ module and then from representation theory we know that all elements of $\gf$ are in integer eigenspaces of $\mathrm{ad}_{\Omega_0^i}$. Hence $\Omega_0^i \in \mathcal{I}$ which is a discrete set.
\end{proof}
Remark: This implies that the $V_{\pm 2}^{\Omega_0^i}$ spaces defined above correspond to genuine $V_2$'s as defined previously for spherically symmetric models. The $V_{\pm 2}^{\Omega_0^i}$  intersect with the $V_2$ associated to $\Lambda_0$. While we have not shown that $2\pi i\Omega_0^i \in \overline{\mathcal{W}}_{\mathbb{R}}$, it must be in some Weyl Chamber, which can always be obtained from the fundamental one by discrete reflections.
\begin{proposition}
The set of representative $\Omega_+^i$ is discrete. 
\end{proposition}
\begin{proof}
We have shown in the previous proposition that the set of possible $\Omega_0^i$ is discrete, so it suffices to show that any representative $\Omega_+^i$ and $\Omega_+^j$ that define the same $\Omega_0^i$ are equivalent. To see this consider the equation
\begin{align}
[X_+,X_-] = \Omega_0^i,\quad\quad \text{on } V_2^{\Lambda_0} \cap V_2^{\Omega_0^i}.
\end{align}
Since this equation is on $V_2^{\Lambda_0}$, it follows from Theorem \ref{commdeterminesthm} that all solutions are on the same orbit. Any arbitrary choice of one of these equivalent solutions can then be made for the representative $\Omega_+^i$.
\end{proof}
Remark: There may be solutions to $[X_+,X_-] = \Omega_0^i$ in $V_2^{\Omega_0^i}$ that are not $G_0^{\Lambda_0}$-equivalent, but all of the solutions that are \fbi{also} in $V_2^{\Lambda_0}$ have to be.
\par
We now have the decomposition
\begin{align}
V_2 = V_2\cap V_2^{\Omega_0} \oplus V_2\backslash V_2^{\Omega_0}.
\end{align}
In light of the above propositions, we can now refine the asymptotic property in \eqref{OKasymptote}, since $\mathfrak{F}^{\times}$ is given by the disjoint union,
\begin{align}
\mathfrak{F}^{\times} = \bigsqcup_{i}\mathfrak{F}_i,
\end{align}
where $\mathfrak{F}_i := \left\{ g\cdot\Omega_+^i \,\, | \,\, g\in G_0^{\Lambda_0}\right\}$. 
\par
Combining this result with \eqref{OKasymptote} implies that each bounded solution to the equations will have the asymptotic property
\begin{align}
\left\|\Lambda_+(r) - \mathfrak{F}_i \right\| \rightarrow 0 \quad \text{as} \quad r \rightarrow \infty, {\text{for some }} \mathfrak{F}_i.
\end{align}
The total magnetic charge will only be zero in the case where $\mathfrak{F}_i$ is $\{g\cdot\Omega_+^i\}$ for the $\Omega_+^i$ that corresponds to $\Omega_0^i = \Lambda_0$.
\section{The Dynamical System}
Since the orbits $\mathfrak{F}_i$ are disjoint, we can consider the solutions that asymptotically approach each distinct $\mathfrak{F}_i$ separately (and will now suppress the $i$ index on $\Omega_+^i$). 
%
%If we knew that $\Lambda_+(r)$ limited to a particular point on $\mathfrak{F}_i^{\times}$ we could perform a constant gauge change to ensure %that $\lim_{r\rightarrow\infty} \Lambda_+ = \Omega_+$. In light of this possibility, consider the change of variables
%
%\begin {align}
%Z_+ = \Lambda_+ - \Omega_+
%\end{align}
%
By introducing the variable $\tau$ via
\begin{align}
\frac{dr}{d\tau} = r\sqrt{N},
\end{align}
and letting $\dot{y}$ denote  the derivative of some variable $y$ with respect to $\tau$, we can write the system of ordinary differential equations \eqref{meqn}-\eqref{ym2} in autonomous form as
\begin{align}\label{dynsys}
\dot{z} &= -z -z\nu,\\
\dot{\nu} &= -\nu -\frac{1}{2}\nu^2 -z^2(\breve{G}+P),\\\label{ym2fdelta}
\ddot{\Lambda}_+ &= \dot{\Lambda}_+ -\mathcal{F}(\Lambda_+) + \delta(z,\nu,\Lambda_+,\dot{\Lambda}_+)\dot{\Lambda}_+,
\end{align}
where
\begin{align*}
z &:= r^{-1},  \quad \quad \nu :=\sqrt{N}-1, \quad \quad \breve{G} := \frac{1}{2}\left\|\dot{\Lambda}_+\right\|^2,\\
P &:= \frac{1}{8}\left\|\Lambda_0-[\Lambda_+,\Lambda_-]\right\|^2,\\
%P &:= \frac{1}{8}\left\|[Z_+,Z_-]+[\Omega_+,Z_-]+[Z_+,\Omega_-]\right\|^2 \\
%R(Z_+) &= \frac{1}{2}\left([[Z_+,Z_-],Z_+]+[[Z_+,Z_-],\Omega_+]+[[Z_+,\Omega_-],Z_+]+[[\Omega_+,Z_-],Z_+]\right)\\
\delta(z,\nu,Z_+,\Gamma_+) &:= -1-2\nu +(\nu+1)^{-1}(1+\nu + \frac{1}{2}\nu^2 + z^2(\breve{G}-P)).
\end{align*}
Note that the variables introduced above follow from those originally introduced in \cite{BFM}. If possible, we would like to state the well-posedness of the asymptotic solutions by applying the following existence and uniqueness result, which is a trivial modification of Lemma 3 from \cite{KirchgraberandPalmer}:
\begin{lemma}\label{dslemma3}
Let $B$ be a real $n\times n$ matrix, the eigenvalues of which have nonzero real parts, and the $n$ eigenvectors of which are all distinct. Suppose we have a projection, $\mathbb{P}$, commuting with $B$, and constants, $K>0, \alpha>0$ exist such that, for all $t \geq 0$
\begin{align*}
|e^{tB}\mathbb{P}| \leq Ke^{-\alpha t},\\
|e^{-tB}(\mathbb{I}-\mathbb{P})| \leq Ke^{-\alpha t}.
\end{align*}
Let $h: [0,\infty) \times \mathbb{R}^n\rightarrow\mathbb{R}^n$ be a continuous function, satisfying 
\begin{align*}
|h(t,0)|&\leq \mu,&\quad \text{for all }  t\geq 0,\\
|h(t,y_1)-h(t,y_2)| &\leq l|y_1-y_2|, &\quad \text{for all } y_1,y_2, t\geq 0,
\end{align*}
where $\mu, l$ are positive constants. Then if $l<\frac{\alpha}{2K}$, for any choice of $y_-(q):=\mathbb{P}y(q)$ at some $q\in[0,\infty)$, a bounded solution of the nonlinear equation 
\begin{align}
\dot{y} = By + h(t,y(t))
\end{align}
exists on $[0,\infty)$ and is unique up to the choice of $y_-(q)$. Moreover for all $t\geq 0$,
\begin{align*}
|y(t)| \leq \frac{2K\mu + \alpha K e^{\alpha q}\left|y_-(q)\right|}{\alpha-2Kl}.
\end{align*}
\end{lemma}
By substituting for the right hand side, we see that equations \eqref{dynsys}-\eqref{ym2fdelta} have a critical point at $(z,\nu,\Lambda_+, \dot{\Lambda}_+) = (0,0,\Omega_+,0)$. If we write these equations as a first order dynamical system, the linearized system at this critical point is
\begin{align}
\frac{d}{d\tau}\left(\begin{array}{c}z\\\nu\\\Lambda_+\\\dot{\Lambda}_+\end{array}\right) = \left[\begin{array}{cccc} -1 & 0 & 0 & 0\\ 0 & -1 & 0 & 0 \\ 0 & 0 & 0 & \mathbb{I}\\ 0 & 0 & A & \mathbb{I}\end{array}\right]\left(\begin{array}{c}z\\\nu\\\Lambda_+\\\dot{\Lambda}_+\end{array}\right).
\end{align}
where $A$, the linearization of $-\mathcal{F}$, is
\begin{align}
A &:=\frac{1}{2}\mathrm{ad}_{\Omega_+}\circ(\mathrm{ad}_{\Omega_-}+\mathrm{ad}_{\Omega_+}\circ c)+\frac{1}{2}(\mathrm{ad}_{\Omega_0}-\mathrm{ad}_{\Lambda_0}).
\end{align}
To know whether we can apply Lemma \ref{dslemma3}, we will need a description of the eigenvalues of the linearization. These depend on the eigenvalues of $A$. 
\section{Eigenbasis of $V_2$ for the operator $A$}
\begin{theorem}
$V_2$ has an eigenbasis for the operator $A$ which can be constructed from a set of highest weights.
\end{theorem}
We will prove this result after establishing the necessary propositions and lemmata.
\begin{proposition}
$A$ is symmetric on $\gf$.
\end{proposition}
\begin{proof}
Let $X, Y$ be arbitrary elements of $\gf$. Then
\begin{align*}
\brac X | A Y \ket &= \frac12 \brac X | [\Omega_+, [\Omega_-,Y] + [\Omega_+,c(Y)]] + [\Omega_0 - \Lambda_0,Y] \ket\\
									 &= \frac12 \brac X | [\Omega_+, [\Omega_-,Y] + [\Omega_+,c(Y)]]\ket + \frac12 \brac X |[\Omega_0 - \Lambda_0,Y] \ket\\
									 &= \frac12 \brac [-c(\Omega_+),X] | [\Omega_-,Y] + [\Omega_+,c(Y)]\ket + \frac12 \brac [-c(\Omega_0 - \Lambda_0), X] | Y \ket\\
									 &= \frac12 \brac [\Omega_+,[\Omega_-,X]] | Y \ket  + \frac12 \brac  [\Omega_-, [\Omega_-, X]] | c(Y) \ket + \frac12 \brac [\Omega_0 - \Lambda_0, X] | Y \ket\\
									 &= \frac12 \brac [\Omega_+,[\Omega_-,X]] | Y \ket  + \frac12 \brac  [\Omega_+, [\Omega_+, X]] | Y \ket + \frac12 \brac [\Omega_0 - \Lambda_0, X] | Y \ket\\
									 &= \brac A X | Y \ket
\end{align*}
\end{proof}
Since the operator $A$ is made up of elements of the standard triple $\left\{{\Omega_0,\Omega_+,\Omega_-}\right\}$ we consider the $\slf_2\C$ subalgebra of $\mathfrak{g}$ defined by the triple together with the induced Lie bracket. Accordingly, we can consider $\mathfrak{g}$ as an $\slf_2$-module and write it in terms of highest weights of $\mathrm{ad}_{\Omega_0}$ and the lowering operator $\mathrm{ad}_{\Omega_-}$. 
\begin{proposition}
Consider $\gf$ as an $\slf_2$-module corresponding to the subalgebra defined by the standard triple $\left\{{\Omega_0,\Omega_+,\Omega_-}\right\}$ . A basis can be chosen for the highest weights such that they are also $\Lambda_0$-weights.
\end{proposition}
\begin{proof}
The key fact is that since $\Omega_+$ is in $V_2$, it is a $\Lambda_0$-weight. Then consider a highest weight $\mu$ such that
\begin{align}
\mu = \bigoplus_{i,j}\mu_{i,j},  \text{ where } \mu_{i,j}\in V(i,j),
\end{align}
and $V(i,j)$ is defined to be the $(i,j)$ eigenspace of $(\Lambda_0,\Omega_0)$ (which may be empty). Then $\Omega_+ \in V(2,2)$ and
\begin{align} 
\mathrm{ad}_{\Omega_+} : V(i,j) \rightarrow V(i+2,j+2).
\end{align} 
This means that each $\mu_{i,j}$ will have to independently be a highest weight as the equation $\mathrm{ad}_{\Omega_+}\mu=0$ will have to be solved independently in each $V(i+2,j+2)$.
\par
Accordingly we can assume each highest weight of $\mathrm{ad}_{\Omega_0}$ belongs to only one $V(i,j)$.
\end{proof}
Under the action of the lowering operator $\mathrm{ad}_{\Omega_-}$ the image of these highest weights then span all of $\mathfrak{g}$. In particular, if $\{ \mu_{m,2k}^a \}$ is a set of highest weights of the $\Omega-\slf_2$ decomposition with even $\Lambda_0$-weight which satisfy $m\geq k-1,\, k\geq1$ (the $a$ index is over highest weights with the same $(m,2k)$), then 
\begin{align}
\left\{ \Omega_{-}^{k-1}\cdot\mu_{m,2k}^a \right\}
\end{align}\index{$\mu_{m,2k}^a$}\index{$\Omega_-^{k-1}\cdot\mu_{m,2k}^a$}
is a basis for $V_2$ over $\C$.
\par
Note: To simplify the notation we will from here on use the convention of denoting successive applications of $\mathrm{ad}_X$ by the corresponding power of $\mathrm{ad}_X$ and, where suitable, we will also use $X \cdot Y$ to denote $\mathrm{ad}_X Y$. 
\par
To obtain an eigenbasis of $V_2$ for the operator $A$ we now establish some properties of the action of $c\, \circ  \mathrm{ad}^2_{\Omega_-}$ on this basis. We will need the following formula from $\slf_2$ representation theory \cite{Knapp}:
\begin{proposition}
If $\mu_{x}$ is a highest weight in a representation of an $\slf_2$ $\{h,e,f\}$ with weight $x$ and 
\begin{align}
\psi(a,b,\mu_{x}) = e^a \cdot f^b\cdot \mu_{x},\quad b < x,
\end{align}
then $\psi$ has the recursive property
\begin{align}\label{efhrecursive}
\psi(a,b,\mu_x) = ((x+1)b-b^2)\psi(a-1,b-1,\mu_x).
\end{align}
\end{proposition}
We can now prove the following:
\begin{proposition}
If $m \geq k+1$ then $c(\Omega_-^{k+1}\cdot\mu_{m,2k}^a)$ will lie in $V_2$ and be nonzero. If $m = k$ or  $m=k-1$ then $c(\Omega_-^{k+1}\cdot\mu_{m,2k}^a)$ is zero.
\end{proposition}
\begin{proof}
From $\slf_2$ theory we know that if $m$ is the highest weight of an irreducible $\slf_2$ representation then $-m$ is the corresponding lowest weight. The map $\mathrm{ad}_{\Omega_-}$ lowers the weight by 2, so after applying $\mathrm{ad}_{\Omega_-}^{k-1}$ we will have moved down the $\slf_2$ string to $m-2k+2$ and can only apply the lowering operator a further $m-k+1$ times before getting zero. This means that applying the lowering operator to the highest weight $k+1$ times will only give something nonzero if $m-k+1 \geq 2$. This proves the second part of the proposition. When $m\geq k+1$ we can prove the first case of the proposition by applying the identity \eqref{efhrecursive}.
\end{proof}
\begin{proposition}
$\left.c\circ \Omega_-^2\circ c \circ \Omega_-^2\right|_{V_2}$ is diagonal on the above basis.
\end{proposition}
\begin{proof}
Firstly, if $m =k$ or $m=k-1$, then the element $\mu_{m,2k}^a$ maps to zero which is trivially a diagonal map. For the remainder of the proof consider $m \geq k+1$ 
By applying \eqref{efhrecursive}, it follows that 
\begin{align*}
c\cdot \Omega_-^2\cdot c \cdot \Omega_-^2 \cdot \Omega_-^{k-1} \mu_{m,2k}^a &= c \cdot c \cdot \Omega_+^2 \cdot \Omega_-^{k+1} \mu_{m,2k}^a\\
&= \Omega_+^2 \cdot \Omega_-^{k+1} \mu_{m,2k}^a\\
&= ((m+1)(k+1)-(k+1)^2)\Omega_+ \cdot \Omega_-^{k} \mu_{m,2k}^a\\
&= (m-k)(k+1)\Omega_+ \cdot \Omega_-^{k} \mu_{m,2k}^a\\
&= (m-k)(k+1)((m+1)k-k^2) \Omega_-^{k-1} \mu_{m,2k}^a\\
&= (m-k)(k+1)(m+1-k)k \Omega_-^{k-1} \mu_{m,2k}^a.
\end{align*}
Therefore we have a diagonal map. The coefficient is nonzero when $m \geq k+1$. When the two elements $\Omega_-^{k-1} \mu_{m,k}^a,c(\Omega_-^{k+1} \mu_{m,k}^a)$ are not proportional, the map $c \circ \Omega_-^2$ is a nondegenerate linear transformation of the pair. 
\end{proof}
\par
We use the above propositions to define a new basis for $V_2$. Firstly the new basis contains all of the elements of $\{\Omega_{-}^{k-1}\mu_{m,2k}^a\}$ that satisfy $m=k$ or $m=k-1$.
\par
Therefore the rest of the elements in the original basis satisfy $m\geq k+1$. Consider the image of such an element, $\Omega_{-}^{k-1}\mu_{m,2k}^a$, under $c \circ \Omega_-^2$. Since
\begin{align*}
\mu_{m,2k}^a \in V(m,2k),
\end{align*}
it follows that 
\begin{align*}
\Omega_{-}^{k-1}\mu_{m,2k}^a &\in V(m - 2(k-1),2k - 2(k-1))\\
&= V(m - 2k+2,2).
\end{align*}
Hence
\begin{align*}
c \cdot \Omega_-^2: V(m-2k+2,2) \rightarrow V(-m+2k+2,2)
\end{align*}
and then because
\begin{align*}
m-2k+2 = -m+2k + 2 \implies m = 2k,
\end{align*}
we have that if $m=2k$ it is possible that $\Omega_{-}^{k-1}\mu_{m,2k}^a$ maps to a multiple of itself under $c \circ \Omega_-^2$. We can take the highest weights satisfying $m=2k$ and restrict the basis defined by Oliynyk and K\"unzle for $V_2^{\Omega_0}$ to $V_2^{\Omega_0}\cap V_2^{\Lambda_0}$, i.e. we define as in \cite{OK02b} 
\begin{align}
\xi_{2k,2k}^a:= \left\{\begin{array}{ll}
i\mu_{2k,2k}^a + c\left(\frac{i}{(2k)!}\Omega_-^{2k}\cdot\mu_{2k,2k}^a\right) & \text{ if } c\left(\frac{1}{(2k)!}\Omega_-^{2k}\cdot\mu_{2k,2k}^a\right) = - \mu_{2k,2k}^a,\\
\mu_{2k,2k}^a + c\left(\frac{1}{(2k)!}\Omega_-^{2k}\cdot\mu_{2k,2k}^a\right) & \text{ otherwise }\\
\end{array}\right.
\end{align}
and then add $v_{2k,2k} = \Omega_-^{k-1}\xi_{2k,2k}^a$ to the new basis of $V_2^{\Lambda_0}$. In \cite{OK02b}, $\mathcal{E}$ is defined as the $k$ value for the set of all possible even $\Omega_0$ weights with $k \geq 1$ and we will define $\mathcal{E}^{\Lambda_0}$\index{$\mathcal{E}^{\Lambda_0}$} as the subset of $\mathcal{E}$ for which $\Omega_-^{k-1}\mu_{2k,2k}^a$ is in $V_2^{\Lambda_0}$
\par
For the remaining elements in the original basis, $m > k+1$, and $ k \geq 1$, and the $\left\{\Omega_{-}^{k-1}\mu_{m,2k}^a\right\}$ can be arranged into the pairs $\Omega_{-}^{k_i-1}\mu_{m_i,2k_i}^{a_i}$ and $c(\Omega_{-}^{k_i+1}\mu_{m_i,2k_i}^{a_i})$. To see that for any $i$ we can always write 
\begin{align}
c(\Omega_{-}^{k_i+1}\mu_{m_i,2k_i}^{a_i}) = \Omega_{-}^{k_j-1}\mu_{m_j,2k_j}^{a_j}
\end{align}
for some $j$, we apply $\mathrm{ad}_{\Omega_+}^{m-k_i-1}$ and $\mathrm{ad}_{\Omega_+}^{m-k_i}$ to the left hand side to obtain
\begin{align*}
\Omega_+^{m_i-k_i-1}c(\Omega_{-}^{k_i+1}\mu_{m_i,2k_i}^{a_i}) &= (-1)^{m_i-k_i-1}c(\Omega_{-}^{m_i}\mu_{m_i,2k_i}^{a_i}),\\
\,&\,\\
\Omega_+^{m_i-k_i}c(\Omega_{-}^{k_i+1}\mu_{m_i,2k_i}^{a_i}) &= (-1)^{m_i-k_i}c(\Omega_{-}^{m_i+1}\mu_{m_i,2k_i}^{a_i})\\
&=0.
\end{align*}
So $\Omega_+^{m_i-k_i-1}c(\Omega_{-}^{k_i+1}\mu_{m_i,2k_i}^{a_i})$ is a highest weight which can be lowered $m_i-k_i-1+k_i+1 = m_i$ times and is in $V_2$ after being lowered $m_i-k_i-1$ times. 
This shows that for each highest weight $\mu_{m_i,2k_i}^{a_i}$ we can find a highest weight $\mu_{m_i,2(m_i-k_i)}^{a_j}$ such that 
\begin{align}
c(\Omega_{-}^{k_i+1}\mu_{m_i,2(m_i-k_i)}^{a_i}) = \Omega_{-}^{m_i - k_i-1}\mu_{m_i,2(m_i-k_i)}^{a_j}.
\end{align}
From these pairs we will take, as a labelling convention, the remaining elements in the original basis $\left\{\Omega_{-}^{k-1}\mu_{m,2k}^a\right\}$ which satisfy $k_i < m_i-k_i$ and then find the partner element as above.
\par
To summarise, we have refined the basis for $V_2$,  $\left\{\Omega_{-}^{k-1}\mu_{m,2k}^a\right\}$, to 
\begin{align}
\left\{
\begin{array}{lll}
v_{k-1,2k}^a &= \Omega_{-}^{k-1}\cdot\mu_{k-1,2k}^a,& \text{for } m = k-1\\
v_{k,2k}^a &= \Omega_{-}^{k-1}\cdot\mu_{k,2k}^a,& \text{for } m = k \\
v_{2k,2k}^a &= \Omega_{-}^{k-1}\cdot\left(i\mu_{2k,2k}^a + c\left(\frac{i}{(2k)!}\Omega_-^{2k}\cdot\mu_{2k,2k}^a\right)\right) & \text{for $m=2k$ if `C1'}\\
v_{2k,2k}^a &= \Omega_{-}^{k-1}\cdot\left(\mu_{2k,2k}^a + c\left(\frac{1}{(2k)!}\Omega_-^{2k}\cdot\mu_{2k,2k}^a\right)\right) & \text{for $m=2k$ otherwise }\\
{v}_{m,2k}^{a} &= \Omega_{-}^{k-1}\cdot\mu_{m,2k}^a, &\text{for }  m >2k\\
{\breve{v}}_{m,2k}^{a} &= c(\Omega_{-}^{k+1}\cdot\mu_{m,2k}^a), &\text{for }  m > 2k.
\end{array}
\right\}
\end{align}
where `C1' stands for the too-wide condition:
\begin{align*}
c\left(\frac{1}{(2k)!}\Omega_-^{2k}\cdot\mu_{2k,2k}^a\right) = - \mu_{2k,2k}^a.
\end{align*}
We now establish an eigenbasis for $A$ on $V_2^{\Lambda_0}$.  When restricted to the intersection  $V_2^{\Omega_0} \cap V_2^{\Lambda_0}$, our operator becomes identical to the operator in \cite{OK02b}, also called $A$. Hence we can define the same vectors for the eigenbasis of $V_2^{\Omega_0} \cap V_2^{\Lambda_0}$ as
\begin{align}
X_{2k}^a = \left\{\begin{array}{ll}
\frac{1}{(k-1)!}v_{2k,2k}^a & \text{ if $k$ is odd}\\
\frac{i}{(k-1)!}v_{2k,2k}^a & \text{ if $k$ is even}\end{array} \text{  and  $Y_{2k}^a = iX_{2k}^a$} \right\},
\end{align}
with
\begin{align}
\begin{array}{ccc}
A(X_{2k}^a) = k(k+1)X^a_{2k} &\text{ and } &A(Y_{2k}^a) = 0.
\end{array}
\end{align}
For the complement, $V_2^{\Lambda_0}\backslash V_2^{\Omega_0}$, we will consider the different possibilities for $m$ in cases.
\par
When $m = k-1$, we consider the action of $A$ on $v_{k-1,2k}^a$ 
\begin{align*}
A \cdot \Omega_-^{k-1} \cdot \mu_{k-1,2k}^a &= \frac{1}{2}\left(\Omega_+ \cdot \Omega_- + c \cdot \Omega_-^2 + (\Omega_0 - \Lambda_0)\cdot \right)\Omega_-^{k-1} \cdot \mu_{k-1,2k}^a\\
&=\frac{-(k+1)}{2}\Omega_-^{k-1} \cdot \mu_{k-1,2k}^a.
\end{align*}
When $m = k$, we have 
\begin{align*}
A \cdot \Omega_-^{k-1} \cdot \mu_{k,2k}^a &= \frac{1}{2}\left(\Omega_+ \cdot \Omega_- + c \cdot \Omega_-^2 + (\Omega_0 - \Lambda_0)\cdot \right)\Omega_-^{k-1} \cdot \mu_{k,2k}^a\\
&=0.
\end{align*}
The calculations for $iv_{k-1,2k}^a$ and $iv_{k,2k}^a$ are the same as above.
\par
When $m\geq k+1, m\neq 2k$ we consider the image of the pair $v_{m,2k}^a, \breve{v}_{m,2k}^a$
\begin{align*}
A \cdot v_{m,2k}^a =&\frac{1}{2}\left(\Omega_+ \cdot \Omega_- \cdot + c \cdot \Omega_-^2 \cdot +(\Omega_0-\Lambda_0)\cdot\right)\Omega_-^{k-1} \cdot \mu_{m,2k}^a\\\yesnumber
=&\frac{(m-k)(k+1)}{2}v_{m,2k}^a +\frac{1}{2} \breve{v}_{m,2k}^a,\\
&\\
A \cdot \breve{v}_{m,2k}^a =&\frac{1}{2}\left(\Omega_+ \cdot \Omega_- \cdot + c \cdot \Omega_-^2 \cdot +(\Omega_0-\Lambda_0)\cdot\right)c(\Omega_-^{k+1} \cdot \mu_{m,2k}^a)\\\yesnumber
=&\frac{(m+1-k)k}{2} \breve{v}_{m,2k}^a + \frac{(m-k)(k+1)(m+1-k)k}{2} v_{m,2k}^a.
\end{align*}
So the action of $A$ on the pair $v_{m,2k}^a, \breve{v}_{m,2k}^a$ is given by the two-by-two matrix
\begin{align}
&\frac{1}{2}\left[
\begin{array}{ccc}
{(m-k)(k+1)} & & {(m-k)(k+1)(m+1-k)k}\\
& & \\
{1} & & {(m+1-k)k}
\end{array}
\right],
\end{align}
with a similar action on the pair $iv_{m,2k}^a, i\breve{v}_{m,2k}^a$.
By solving these matrices for the eigenvectors and eigenvalues we define
\begin{align*}
X_{m,2k}^a := &\,(k+1)(m-k)\Omega_-^{k-1} \cdot \mu_{m,2k}^a +  c(\Omega_-^{k+1} \cdot \mu_{m,2k}^a),\\
\breve{X}_{m,2k}^a := &\,i(k+1)(m-k)\Omega_-^{k-1} \cdot \mu_{m,2k}^a - i c(\Omega_-^{k+1} \cdot \mu_{m,2k}^a),\\
Y_{m,2k}^a := &\,k(k-m-1)\Omega_-^{k-1} \cdot \mu_{m,2k}^a + c(\Omega_-^{k+1} \cdot \mu_{m,2k}^a),\\
\breve{Y}_{m,2k}^a := &\,ik(k-m-1)\Omega_-^{k-1} \cdot \mu_{m,2k}^a - i c(\Omega_-^{k+1} \cdot \mu_{m,2k}^a),
\end{align*}
satisfying
\begin{align*}
A X_{m,2k}^a &= (km-k^2 +\frac{m}{2})X_{m,2k}^a,\\
A \breve{X}_{m,2k}^a &= (km-k^2 +\frac{m}{2})\breve{X}_{m,2k}^a,\\
A Y_{m,2k}^a &= 0,\\
A \breve{Y}_{m,2k}^a &= 0.
\end{align*}
To summarise, we have an $A$-eigenbasis of $V_2$, $\{X,\breve{X},Y,\breve{Y}\}$, made up of 
\begin{align}\label{eigenbasis}\index{eigenbasis for $A$}
\begin{array}{lc}
\textbf{eigenvector} & \textbf{eigenvalue}\\
\hline
X_{k-1, 2k}^a & -\frac{k+1}{2},\\
\breve{X}_{k-1, 2k}^a & -\frac{k+1}{2},\\
Y_{k,2k}^a & 0,\\
\breve{Y}_{k,2k}^a & 0,\\
X_{2k,2k}^a & k(k+1),\\
Y_{2k,2k}^a & 0,\\
X_{m,2k}^a & mk-k^2+\frac{m}{2},\\
\breve{X}_{m,2k}^a & mk-k^2+\frac{m}{2},\\
Y_{m,2k}^a & 0,\\
\breve{Y}_{m,2k}^a & 0.
\end{array}
\end{align}
The zero eigenvalues of $A$ mean that \eqref{dynsys} does not meet the preconditions of Lemma \ref{dslemma3}. We will look for a way around this over the next few sections.
\section{The First-Order Yang-Mills Equation defines a Horizontal Subspace}
As shown in \cite{OK02b}, if \eqref{ym1} is satisfied at some value of $r$, then it will hold for the same interval of existence as a solution of the rest of the equations. It turns out that we can make further use of \eqref{ym1} by interpreting it geometrically.
\\
If we think of $\Lambda_+$ as an arbitrary curve in $V_2$ and $\Lambda_+'$ as an arbitrary tangent vector, then the left hand side of  \eqref{ym1} can be viewed as defining a map on the tangent bundle $TV_2$, i.e.,
\begin{align}
\mathrm{YM1}: TV_2 \rightarrow \mathfrak{g}_0^{\Lambda_0},\quad \mathrm{YM1}(\Lambda_+,\Lambda_+') := [\Lambda_+,c(\Lambda_+')]+[c(\Lambda_+),\Lambda_+'].
\end{align}
We will now use this map to establish a horizontal and vertical decomposition of the tangent space.
\begin{definition} The vertical space at $\Lambda_+$, is defined as
\begin{align}
\mathrm{vert}_{\Lambda_+} :=\{ \Gamma_+ \in T_{\Lambda_+}V_2 | \Gamma_+ \text{is tangent to a } G_0^{\Lambda_0} - orbit\}.
\end{align}
\end{definition}
\begin{proposition}
\begin{align}
\mathrm{vert}_{\Lambda_+} = \{ [a,\Lambda_+], a \in \gf_0^{\Lambda_0} \}.
\end{align}
\end{proposition}
\begin{proof}
Without loss of generality, consider curves through the identity element of $G_0^{\Lambda_0}$ of the form $g(s) = \exp{(s a)}$. Then
$$\frac{d}{ds}\left.\left(\mathrm{Ad}_{\exp{(s a)}}\Lambda_+\right)\right|_{s=0} = [a,\Lambda_+].$$
\end{proof}
\begin{definition} The horizontal space at $\Lambda_+$ is defined as
\begin{align}\index{$\mathrm{hor}_{\Lambda_+}$}
\mathrm{hor}_{\Lambda_+} := {\mathrm{vert}_{\Lambda_+}}^{\perp},
\end{align} 
where the $\perp$ is calculated with respect to the $\brac \cdot | \cdot \ket$ inner product.
\end{definition}
Therefore
\begin{align}
T_{\Lambda_+}V_2 &= \mathrm{hor}_{\Lambda_+} \oplus \mathrm{vert}_{\Lambda_+},\,\, \forall \Lambda_+\in V_2.
\end{align}
We also define the relevant bundles,
\begin{align}
\mathrm{vert} &:= \bigcup_{\Lambda_+\in V_2} \mathrm{vert}_{\Lambda_+},\\\nonumber
\mathrm{hor} &:= \bigcup_{\Lambda_+\in V_2} \mathrm{hor}_{\Lambda_+},
\end{align}
so that
\begin{align}\index{$\mathrm{hor}$}\index{$\mathrm{vert}$}
{TV}_2 = \mathrm{hor} \oplus \mathrm{vert}.
\end{align} 
\begin{lemma}
For each $\Lambda_+ \in V_2$,
\begin{align}
\mathrm{hor}_{\Lambda_+}=\left\{\Gamma_+\in T_{\Lambda_+}V_2 \,\,|\,\, \mathrm{YM1}(\Lambda_+,\Gamma_+)=0\right\}.
\end{align}
\end{lemma}
\begin{proof}
Consider the inner product of an arbitrary vertical vector, $[a,\Lambda_+] \in \mathrm{vert}_{\Lambda_+}$ with an arbitrary tangent vector $\Gamma_+ \in T_{\Lambda_+}V_2$,
\begin{align*}
\brac [a,\Lambda_+] | \Gamma_+ \ket &= \brac a | [\Gamma_+, -c(\Lambda_+)] \ket\\
&= \frac12 \left( \brac a | [\Gamma_+, -c(\Lambda_+)] \ket + \brac a | [\Gamma_+, -c(\Lambda_+)] \ket\right)\\
&= \frac12 \left( \brac a | [\Gamma_+, -c(\Lambda_+)] \ket + \brac c(a) | [c(\Gamma_+), -\Lambda_+] \ket\right)\\
&= \frac12 \left( \brac a | [c(\Lambda_+),\Gamma_+] \ket + \brac a | [\Lambda_+, c(\Gamma_+)] \ket\right)\\
&= \frac12 \brac a | \mathrm{YM1}(\Lambda_+,\Gamma_+) \ket.
\end{align*}
Then if $\mathrm{YM1}(\Lambda_+,\Gamma_+)=0$ it is clear that $\Gamma_+ \in \mathrm{hor}_{\Lambda_+}$. For the other inclusion, set $a = \mathrm{YM1}(\Lambda_+,\Gamma_+)$ and the result follows.
\end{proof}
Remark: We see that the first-order Yang-Mills equation is equivalent to the connection on $\pi: V_2 \rightarrow V_2/G_0^{\Lambda_0}$ that is determined by the $\mathrm{hor}\oplus\mathrm{vert}$ splitting. We can therefore define a curvature on each manifold making up strata in $V_2/G_0^{\Lambda_0}$.
Following \cite{BKMM96}, we can view the splitting as an Ehresmann connection by defining
The curvature $\mathcal{B}: TV_2 \times TV_2 \rightarrow \mathrm{vert}$,  is completely determined by the map $\tilde{B}:TV_2 \times TV_2 \rightarrow \mathfrak{g}_0^{\Lambda_3}$,
\begin{align}\label{YM1curvaturedefn}
&\tilde{B}((\Lambda_+,\Lambda_+'),(\Lambda_+,\Sigma_+')) = \mathrm{YM1}(\mathbb{H}_{\Lambda_+}\Lambda_+', \mathbb{H}_{\Lambda_+}\Sigma_+'),
\end{align}
$\forall (\Lambda_+,\Lambda_+'),(\Lambda_+,\Sigma_+') \in TV_2.$
The relationship is 
\begin{align}\index{curvature}\index{$\tilde{B}$}
\brac [a,\Lambda_+] ,\mathcal{B}(\Gamma,\Omega) \ket = \brac a | \tilde{B}(\Gamma, \Omega) \ket.
\end{align}
\begin{theorem}
Abelian models are flat.
\end{theorem}
\begin{proof}
As shown in \cite{MFthesis}, for an Abelian model, for any $\Lambda_+$ corresponding to a point on the principal stratum of $V_2/G_0^{\Lambda_0}$ a basis over $\R$ for the horizontal space at $\Lambda_+$ is given by the root vector components of $\Lambda_+$. Then by the triviality of \eqref{ym1} for Abelian models and \eqref{YM1curvaturedefn}, (see for example Equation (3.19) in \cite{Oli02a}), all of the components of $\tilde{\mathcal{B}}$ are zero. 
\end{proof}
\begin{theorem}
Some non-Abelian models are not flat.
\end{theorem}
\begin{proof}
The example of a non-Abelian model given in \eqref{su4nonabelian}, has curvature with nonzero components.
\end{proof}
The full significance of this curvature for the properties of solutions for non-Abelian models, (different behaviour, stability, etc), is yet to be thoroughly explored, but here the splitting of the tangent space into horizontal and vertical spaces provides the key to completing the dynamical system analysis of asymptotic solutions.
\section{Decomposition of the eigenspace of $A$ at $\Omega_+$}\label{ym1eigsofA}
We will use the eigenbasis defined in \eqref{eigenbasis}, to calculate the image of $\mathrm{YM1}$ at $\Omega_+$. 
\begin{proposition}
The `$X$'-type eigenvectors defined in \eqref{eigenbasis} are horizontal at $\Omega_+$, whereas the `$Y$'-type eigenvectors are vertical at $\Omega_+$.
\end{proposition}
\begin{proof}
By calculating, it follows that
\begin{align*}
\mathrm{YM1}(X_{k-1, 2k}^a,\Omega_+)=0,\\
\mathrm{YM1}(\breve{X}_{k-1, 2k}^a,\Omega_+)=0,\\
\mathrm{YM1}(X_{2k,2k}^a,\Omega_+)=0,\\
\mathrm{YM1}(X_{m,2k}^a,\Omega_+)=0,\\
\mathrm{YM1}(\breve{X}_{m,2k}^a,\Omega_+)=0,\\
\end{align*}
which implies that the `$X$'-type eigenvectors are all horizontal at $\Omega_+$, while for the `$Y$'-type eigenvectors, the image under $\mathrm{YM1}$ at $\Omega_+$ is nonzero, and further calculation gives the additional result that
\begin{align*}
[\mathrm{YM1}(Y_{k,2k}^a,\Omega_+),\Omega_+] \propto Y_{k,2k}^a,\\
[\mathrm{YM1}(\breve{Y}_{k,2k}^a,\Omega_+),\Omega_+] \propto \breve{Y}_{k,2k}^a,\\
[\mathrm{YM1}(Y_{2k,2k}^a,\Omega_+),\Omega_+] \propto Y_{2k,2k}^a,\\
[\mathrm{YM1}(Y_{m,2k}^a,\Omega_+),\Omega_+] \propto Y_{m,2k}^a,\\
[\mathrm{YM1}(\breve{Y}_{m,2k}^a,\Omega_+),\Omega_+] \propto \breve{Y}_{m,2k}^a,
\end{align*}
which, since $\mathrm{YM1}(V_2,\Omega_+) \subset \mathfrak{g}_0^{\Lambda_0}$, shows that the zero eigenvalues of the $A$ operator (and the consequent zero eigenvalues in our linearization of the dynamical system) are all associated to vertical vectors at $\Omega_+$.
\end{proof}
We know that the vertical vectors are tangent vectors to curves in the orbits of the residual gauge group, and that the action of the residual gauge group takes solutions to other gauge-equivalent solutions. 
\par
This suggests that we may be able to project the dynamical system onto the space spanned only by the `$X$'-type eigenvectors, where there are no zeroes in the new linearization, then apply the existence theorems above, and then use the first-order Yang-Mills equation to recover the properties of the full solutions. In the next section we introduce some local coordinates in order to achieve this result.
\section{Local Coordinates}\label{sLocCo}
Let us now establish coordinates such that
\begin{align}\label{localcoordinates}
\Lambda_+ = \mathrm{Ad}_g X_+
\end{align}
in a neighbourhood around $\Omega_+$, where $g\in G_0^{\Lambda_0}$ and $X_+\in \mathrm{hor}_{\Omega_+} V_2$.
We know that $\mathrm{vert}_{\Omega_+}V_2 = [\mathfrak{g}^{\Lambda_0}_0,\Omega_+]$, and so denoting the complement to $\mathfrak{g}^{\Lambda_0}_0\cap\mathfrak{g}^{\Omega_+}$ in $\mathfrak{g}^{\Lambda_0}_0$ by $L$, we have the isomorphism
\begin{align}
\mathrm{ad}_{\Omega_+}: L \rightarrow \mathrm{vert}_{\Omega_+}V_2.
\end{align}
Now consider the map
\begin{align}
\Psi : G^{\Lambda_0}_0 \times \mathrm{hor}_{\Omega_+}V_2 \rightarrow V_2, \Psi(g,X_+) = \mathrm{Ad}_{g}X_+.
\end{align}
For a sufficiently small neighbourhood of the identity, every compact Lie group is locally diffeomorphic to its Lie algebra (the real tangent space at the identity). Hence we can introduce logarithmic coordinates, $a\in\mathfrak{g}_0^{\Lambda_0}$, with local diffeomorphism given by the exponential map. Then the differential of 
\begin{align}
\bar{\Psi}: \mathfrak{g}_0^{\Lambda_0}\times \mathrm{hor}_{\Omega_+}V_2 \rightarrow V_2, \bar{\Psi}(a,X_+) = \mathrm{Ad}_{\exp{a}}X_+
\end{align}
evaluated at $(0,\Omega_+)$ is 
\begin{align}
D\bar{\Psi}\left|_{(0,\Omega_+)}(b,V_+)\right. = -\mathrm{ad}_{\Omega_+}b + V_+ .
\end{align}
This map has a kernel from the component of $b$ that lies in $\mathfrak{g}^{\Lambda_0}_0\cap\mathfrak{g}^{\Omega_+}$ so it is not an isomorphism, but it is a surjection, since $\mathrm{ad}_{\Omega_+} \times \mathbb{I}$ is an isomorphism from $L \times \mathrm{hor}_{\Omega_+}V_2$ to $\mathrm{vert}_{\Omega_+}V_2 \oplus \mathrm{hor}_{\Omega_+}V_2 = V_2$. Hence we have for any $C^2$ curve, $\gamma$, in a sufficiently small neighbourhood of $V_2$ around $\Omega_+$ a (non-unique) $C^2$ curve, $\bar{\gamma}$ in $G_0^{\Lambda_0} \times \mathrm{hor}_{\Omega_+} V_2$ such that $\Psi(\bar{\gamma}) = \gamma$.
\section{The Reduced Dynamical System}
In the neighbourhood of $\Omega_+$, where the coordinates $(g,X_+)\in G_0^{\Lambda_0}\times \mathrm{hor}_{\Omega_+}V_2$ such that 
\begin{align*}
\Lambda_+ = \mathrm{Ad}_g X_+
\end{align*}
are well-defined, we differentiate to obtain 
\begin{align}
\Lambda_+' = \mathrm{Ad}_g\left( X_+' + [\eta,X_+] \right),
\end{align}
where $\eta:=g^{-1}g' \in \gf_0^{\Lambda_0}$. In these coordinates, \eqref{ym1} becomes
\begin{align}\label{etaxym1}
[X_+', X_-] + [X_-', X_+] + [[\eta, X_+], X_-] + [[\eta, X_-], X_+] = 0.
\end{align}
\begin{definition}
Let $\mathcal{S}_{X_+}: \gf_0^{\Lambda_0} \rightarrow \gf_0^{\Lambda_0}$ be the operator defined by\index{$\mathcal{S}_{X_+}$}
\begin{align}
\mathcal{S}_{X_+}a = [[a, X_+], X_-] + [[a, X_-], X_+].
\end{align}
\end{definition}
\begin{proposition}
$\mathcal{S}_{X_+}$ is symmetric on $\gf_0^{\Lambda_0}$.
\end{proposition}
\begin{proof}
By using the properties of the inner product, we get
\begin{align*}
\brac \mathcal{S}_{X_+}a | b \ket\\
=&\brac [[a, X_+], X_-] + [[a, X_-], X_+] | b \ket \\
=&\brac a |[[b, X_+], X_-] + [[b, X_-], X_+] \ket \\
=&\brac a | \mathcal{S}_{X_+}b \ket, \text{ for all } a,b \in \gf_0^{\Lambda_0},
\end{align*}
which shows that $\mathcal{S}_{X_+}$ is a symmetric operator on $\gf_0^{\Lambda_0}$ and is therefore diagonalisable.
\end{proof}
\begin{proposition}
The kernel of $\mathcal{S}_{X_+}$ is identical to the kernel of $\mathrm{ad}_{X_+}$.
\end{proposition}
\begin{proof}
Suppose we have some element, $a\in\gf_0^{\Lambda_0}$, such that $\mathcal{S}_{X_+}a = 0$. Then we can take the inner product of $\mathcal{S}_{X_+}a$ with $a$ and use the properties of the inner product to get
\begin{align*}
\brac \mathcal{S}_{X_+}a | a \ket =&\brac [[a, X_+], X_-] + [[a, X_-], X_+] | a \ket \\
=&\brac [[a, X_+], X_-]|a \ket  + \brac [[a, X_-], X_+] | a \ket \\
=&\brac [a, X_+]|[a,X_+] \ket  + \brac [a, X_-] | [a,X_-] \ket \\
=&2\left\| [a, X_+]\right\|^2,
\end{align*}
which implies that the kernel of $\mathcal{S}_{X_+}$ is contained within the kernel of $\mathrm{ad}_{X_+}$, with the reverse inclusion immediate from the definition of $\mathcal{S}_{X_+}$.
\end{proof}
It is then natural to consider the following related operator:
\begin{definition}
Let  $\tilde{\mathcal{S}}_{X_+} : \gf_0^{\Lambda_0}\backslash \mathrm{ker}_{\mathrm{ad}_{X_+}}\gf_0^{\Lambda_0} \rightarrow \gf_0^{\Lambda_0}\backslash \mathrm{ker}_{\mathrm{ad}_{X_+}}\gf_0^{\Lambda_0}$ be the operator defined by
\begin{align}
\tilde{\mathcal{S}}_{X_+} := \mathrm{pr}^{\perp}_{\mathrm{ker}_{\mathrm{ad}_{X_+}}\gf_0^{\Lambda_0}} \circ \mathcal{S}_{X_+} \circ \mathrm{pr}^{\perp}_{\mathrm{ker}_{\mathrm{ad}_{X_+}}\gf_0^{\Lambda_0}}
\end{align}\index{$\tilde{\mathcal{S}}_{X_+}$}
\end{definition}
$\tilde{\mathcal{S}}_{X_+}$ is then an invertible map from $\gf_0^{\Lambda_0}\backslash \mathrm{ker}_{\mathrm{ad}_{X_+}}\gf_0^{\Lambda_0} \rightarrow \gf_0^{\Lambda_0}\backslash \mathrm{ker}_{\mathrm{ad}_{X_+}}\gf_0^{\Lambda_0}$. Let $\eta = \tilde{\eta} + \eta_0$ where $[\eta_0,X_+]=0$ and $\tilde{\eta}:=\mathrm{pr}^{\perp}_{\mathrm{ker}_{\mathrm{ad}_{X_+}}}\eta$. Equation \eqref{etaxym1} can now be written as 
\begin{align}
\mathrm{YM1}(X_+', X_+) + \tilde{\mathcal{S}}_{X_+}\tilde{\eta}  = 0,
\end{align}
and then solved, resulting in
\begin{align}
\tilde{\eta} &= -\tilde{\mathcal{S}}^{-1}_{X_+}\,\mathrm{YM1}(X_+', X_+)\\
&= -\tilde{\mathcal{S}}^{-1}_{X_+}\circ\mathrm{YM1}_{X_+} (X_+'),
\end{align}\index{$\mathrm{YM1}_{X_+}$}
where the last line defines the convenient notation $\mathrm{YM1}_{X_+} (X_+') := \mathrm{YM1}(X_+', X_+)$.
\par
If for some continuous interval on the curve in the local coordinates $(g(r), X_+(r))$, the kernel of  $\mathrm{ad}_{\mathrm{X_+}}$ is nontrivial, then we can let $h$ be the solution to the differential equation
\begin{align}
h' = \eta_0h
\end{align}
for some fixed choice of $h(r_0)$ at a specified point $r_0$ on the interval. We then use the equivalent curve $(q,X_+)$, where $q = gh$ and so consequently satisfies 
\begin{align*}
q^{-1}q' &= h^{-1}(g^{-1}g')h + h^{-1}h'\\
&= h^{-1}(\tilde{\eta} + \eta_0)h + h^{-1}(-\eta_0)h'\\
&= h^{-1}\tilde{\eta}h.
\end{align*}
Then since $[h^{-1}\tilde{\eta}h, X_+] = h^{-1}[\tilde{\eta},X_+]h$, and $h$ is compact, the `new' $\eta=q^{-1}q'$ has trivial kernel under $\mathrm{ad}_{X_+}$. Thus we can fix the choice, $\eta_0 =0$. Hence,

\begin{align}\label{etaasxx}
\eta = -\tilde{\mathcal{S}}^{-1}_{X_+}\circ\mathrm{YM1}_{X_+} (X_+')
\end{align}\index{$\eta$}
We now want to write equation \eqref{ym2fdelta},
\begin{align*}
\ddot{\Lambda}_+ &= \dot{\Lambda}_+ -\mathcal{F}(\Lambda_+) + \delta(z,\nu,\Lambda_+,\dot{\Lambda}_+)\dot{\Lambda}_+,
\end{align*}
in terms of the variables \eqref{localcoordinates}. In these coordinates, the derivatives are
\begin{align}
\dot{\Lambda}_+ &= g \cdot(  \dot{X}_+ + [\eta,X_+]),\\
\ddot{\Lambda}_+ &= g \cdot(  \ddot{X}_+ + [\dot{\eta},X_+]+2[\eta,\dot{X}_+] + [\eta,[\eta,X_+]]),
\end{align}
and
\begin{align*}
\mathcal{F}(\Lambda_+) &= \mathcal{F}(g \cdot X_+)\\
&=g \cdot X_+ -\frac{1}{2} [[g \cdot X_+, g \cdot X_-], g \cdot X_+]\\
&=g \cdot ( X_+ - \frac{1}{2} [[X_+,X_-],X_+] )\\\yesnumber
&=g \cdot \mathcal{F}(X_+).
\end{align*}
The terms in $\delta$ involving $\Lambda_+$ and $\dot{\Lambda}_+$ become
\begin{align*}
P &= \frac{1}{8}\left\|\Lambda_0 - [\Lambda_+,\Lambda_-]\right\|^2\\
&= \frac{1}{8}\left\|\Lambda_0 - [g\cdot X_+, g \cdot X_-]\right\|^2\\
&= \frac{1}{8}\left\|\Lambda_0 - g\cdot[X_+, X_-]\right\|^2\\
&= \frac{1}{8}\left\|g\cdot\Lambda_0 - g\cdot[X_+, X_-]\right\|^2\\\yesnumber
&= \frac{1}{8}\left\|\Lambda_0 - [X_+, X_-]\right\|^2,
\end{align*}
where we have used the fact that $g\cdot \Lambda_0 = \Lambda_0$ (by definition of $G_0^{\Lambda_0}$), and
\begin{align*}
\breve{G} &= \frac{1}{2}\left\|\dot{\Lambda}_+\right\|^2\\
&= \frac{1}{2}\left\|\dot{X}_+ + [\eta,X_+] \right\|^2\\\yesnumber
&= \frac{1}{2}\left\|\dot{X}_+\right\|^2 - \frac{1}{2}\left\|[\eta,X_+]\right\|^2,
\end{align*}
where the last line follows from writing the first-order Yang-Mills equation in the $(g, X_+)$ variables and the properties of the inner product.
\par
Substituting the above expressions into \eqref{ym2fdelta} and applying $\mathrm{Ad}_{g^{-1}}$, we obtain
\begin{align*}
\ddot{X}_+ + [\dot{\eta},X_+] = &\dot{X}_+ - \mathcal{F}(X_+)-2[\eta,\dot{X}_+] - [\eta,[\eta,X_+]] + [\eta,X_+]\\\yesnumber
&+\delta(z,\nu,X_+,\dot{X}_+)\left(\dot{X}_++[\eta,X_+]\right),
\end{align*}
where $\eta$ is determined by $X_+, \dot{X}_+$ and so we substitute the formula \eqref{etaasxx} at each appearance. To obtain an expression for $\dot{\eta}$ we differentiate \eqref{etaxym1}. The result is
\begin{align*}
0 = &[\ddot{X}_+,X_-] + [\ddot{X}_-,X_+] + [[\dot{\eta}, X_+], X_-] + [[\dot{\eta}, X_-], X_+]\\
 &+ [[\eta, \dot{X}_+], X_-] + [[\eta, \dot{X}_-], X_+] + [[\eta, X_+], \dot{X}_-] + [[\eta, X_-], \dot{X}_+].
\end{align*}
In the notation introduced above, we therefore have
\begin{align}
0 = \mathrm{YM1}_{X_+}\ddot{X}_+ + \tilde{\mathcal{S}}_{X_+}\dot{\eta} + 2\mathrm{YM1}_{X_+}[\eta,\dot{X}_+] - \eta\cdot\mathrm{YM1}_{X_+}\dot{X}_+,
\end{align}
which has the solution for $\dot{\eta}$,
\begin{align}
\dot{\eta} = -\tilde{\mathcal{S}}^{-1}_{X_+}\left(\mathrm{YM1}_{X_+}\ddot{X}_+ + 2\mathrm{YM1}_{X_+}[\eta,\dot{X}_+] - \eta\cdot\mathrm{YM1}_{X_+}\dot{X}_+\right).
\end{align}
Hence
\begin{align}\nonumber
\mathrm{pr}_{\mathrm{hor}_{\Omega_+}} [\dot{\eta},X_+]
&=
 \left(\mathrm{pr}_{\mathrm{hor}_{\Omega_+}V_2}\circ \mathrm{ad}_{X_+}\circ\tilde{\mathcal{S}}^{-1}_{X_+}\circ\mathrm{YM1}_{X_+}\right)\ddot{X}_+\\
 &+ \left(\mathrm{pr}_{\mathrm{hor}_{\Omega_+}V_2}\circ \mathrm{ad}_{X_+}\circ\tilde{\mathcal{S}}^{-1}_{X_+}\right)\left( 2\mathrm{YM1}_{X_+}[\eta,\dot{X}_+] - \eta\cdot\mathrm{YM1}_{X_+}\dot{X}_+\right).
\end{align}
\begin{definition}
Let $Q_{X_+} : \mathrm{hor}_{\Omega_+}V_2 \rightarrow \mathrm{hor}_{\Omega_+}V_2$ be the operator defined by 
\begin{align}
Q_{X_+}:= \mathrm{pr}_{\mathrm{hor}_{\Omega_+}V_2}\circ \mathrm{ad}_{X_+}\circ\tilde{\mathcal{S}}^{-1}_{X_+}\circ\mathrm{YM1}_{X_+},
\end{align}
\end{definition}
Then we can write the projected differential equation on $\mathrm{hor}_{\Omega_+}V_2$ as\index{$Q_{X_+}$} 
\begin{align*}
\left(\mathbb{I} + Q_{X_+}\right)\ddot{X}_+ = &\dot{X}_+ - \mathrm{pr}_{\mathrm{hor}_{\Omega_+}V_2}(\mathcal{F}(X_+)) +\delta\dot{X}_+\\
+&\mathrm{pr}_{\mathrm{hor}_{\Omega_+}V_2}\left(
 -2[\eta,\dot{X}_+] - [\eta,[\eta,X_+]] + [\eta,X_+]+\delta[\eta,X_+]\right)\\
-&\left(\mathrm{pr}_{\mathrm{hor}_{\Omega_+}V_2}\circ \mathrm{ad}_{X_+}\circ\tilde{\mathcal{S}}^{-1}_{X_+}\right)\left( 2\mathrm{YM1}_{X_+}[\eta,\dot{X}_+] - \eta\cdot\mathrm{YM1}_{X_+}\dot{X}_+\right).
\end{align*}
We will group all of the terms on the right hand side which have at least one $\eta$ as $J_{\eta}(X_+,\dot{X}_+)$, and write the above equation as 
\begin{align}\label{IQddX}
\left(\mathbb{I} + Q_{X_+}\right)\ddot{X}_+ = &\dot{X}_+ - \mathrm{pr}_{\mathrm{hor}_{\Omega_+}V_2}(\mathcal{F}(X_+)) +\delta\dot{X}_+ + J_{\eta}(X_+,\dot{X}_+).
\end{align}
The linearization of $-\mathcal{F}(X_+)$ at $\Omega_+$ is $A X_+$ and
\begin{align*}
-\mathcal{F}(X_+) &= -X_+ + \frac{1}{2} [[X_+,X_-],X_+]\\
&= -(Z_+ +\Omega_+) + \frac{1}{2} [[Z_++\Omega_+,Z_-+\Omega_-],Z_++\Omega_+]\\
&= -Z_+ -\Omega_+ + AZ_+ + Z_+ +\Omega_+ +\tilde{R}(Z_+)\\
&= AZ_+ + \tilde{R}(Z_+),
\end{align*}
where $Z_+:=X_+ -\Omega_+$\index{$Z_+$} and 
\begin{align*}
\tilde{R}(Z_+) := \frac{1}{2}\left([[Z_+,Z_-],Z_+] + [[Z_+,Z_-],\Omega_+] + [[Z_+,\Omega_-],Z_+] + [[\Omega_+,Z_-],Z_+]\right).
\end{align*}
Let $R(Z_+) := \mathrm{pr}_{\mathrm{hor}_{\Omega_+}V_2}\tilde{R}(Z_+)$ and we can now write \eqref{IQddX} as 
\begin{align}\label{IQddX2}
\left(\mathbb{I} + Q_{X_+}\right)\ddot{Z}_+ = &\dot{Z}_+ + AZ_+ +R(Z_+) +\delta\dot{Z}_+ + J_{\eta}(X_+,\dot{X}_+).
\end{align}
Let $\Gamma_+ := \dot{Z}_+ = \dot{X}_+$. We now have, via \eqref{IQddX2}, the first order system on $\mathbb{R}\times\mathbb{R}\times\mathrm{hor}_{\Omega_+}V_2\times\mathrm{hor}_{\Omega_+}V_2$,
\begin{align}\label{eq753}
\dot{z} &= -z -z\nu,\\
\dot{\nu} &= -\nu -\frac{1}{2}\nu^2 -z^2(\breve{G}+P),\\
\dot{Z}_+ &= \Gamma_+,\\\label{ym2gammajunk}
\left(\mathbb{I} + Q_{X_+}\right)\dot{\Gamma}_+ &= \Gamma_+ + AZ_+ +R(Z_+) +\delta\Gamma_+ + J_{\eta}(X_+,\dot{X}_+).
\end{align}
In these variables, the critical point is at $(z,\nu,Z_+,\Gamma_+) = (0,0,0,0)$, which, since $Z_+ = X_+ - \Omega_+$, corresponds to $Q_{X_+} = Q_{\Omega_+}$ at the critical point. By definition,
\begin{align*}
Q_{\Omega_+}=\mathrm{pr}_{\mathrm{hor}_{\Omega_+}V_2}\circ \mathrm{ad}_{\Omega_+}\circ\tilde{\mathcal{S}}^{-1}_{\Omega_+}\circ\mathrm{YM1}_{\Omega_+}.
\end{align*}
Since $\mathrm{hor}_{\Omega_+}V_2 = \{X \in V_2 | \mathrm{YM1}_{\Omega_+} X = 0\}$, we have 
\begin{align}
Q_{\Omega_+} \left(\mathrm{hor}_{\Omega_+}V_2\right) = 0,
\end{align}
implying that the operator on the left hand side, $\mathbb{I} + Q_{\Omega_+}$, is the identity at the critical point. If we can show that both sides of the equation are $C^1$ functions of $z,\nu, X_+ (\text{or } Z_+), \Gamma_+$, then we can use the Implicit Function Theorem to solve \eqref{ym2gammajunk} locally for $\dot{\Gamma}_+$ near the critical point. The continuity of $Q_{X_+}$ near $\Omega_+$ will also allow us to restrict the neighbourhood as necessary to a neighbourhood where $\left\|Q_{X_+}\right\| < 1$ and then we can expand the inverse in powers of $Q_{X_+}$, and then determine the new linearization of the system. 
\begin{proposition}\label{tenprops}
The following statements are true:
\begin{enumerate}[a.]
\item $(\mathbb{I} + Q_{\Omega_+})W_+ = \mathbb{I}W_+, \forall W_+ \in \mathrm{hor}_{\Omega_+}V_2.$
\item $(\mathbb{I} + Q_{X_+})W_+$ is a continuous function of $(X_+, W_+) \in \mathrm{hor}_{\Omega_+}V_2 \times \mathrm{hor}_{\Omega_+}V_2$ for $X_+$ in a sufficiently small neighbourhood of $\Omega_+$.
\item $\frac{d}{dt}\left((\mathbb{I} + Q_{\Omega_++tX_+})W_+\right)_{t=0} = \mathbb{I}W_+$.
\item $D((\mathbb{I} + Q_{X_+})W_+)\cdot(V_+,M_+) = \frac{d}{dt}\left((\mathbb{I} + Q_{X_++tV_+})(W_++tM_+)\right)_{t=0}$ is a continuous function of $(X_+,W_+)$ for $X_+$ in a sufficiently small neighbourhood of $\Omega_+$, for all $V_+,M_+ \in \mathrm{hor}_{\Omega_+}V_2$.
\item $\eta(\Omega_+,0) = 0$.
\item $\eta : \mathrm{hor}_{\Omega_+}V_2 \times \mathrm{hor}_{\Omega_+}V_2 \rightarrow \gf_0^{\Lambda_0}$ is a continuous function from a neighbourhood of $(\Omega_+,0)$ to a neighbourhood of $(0,0)$.
\item $D\eta(\Omega_+,0)\cdot(V_+,M_+) = 0$.
\item $R(0) = 0$.
\item $DR(0) = 0$.
\item $\delta(0,0,0,0) = 0$.
\end{enumerate}
\end{proposition}
\begin{proof}
These statements follow from the definitions of each function. Commutators can be expanded over a constant basis, e.g.,
\begin{align*}
[A,B] = A_iB_j[T_i,T_j]
\end{align*}
and then a constant bound on $\left\|[T_i,T_j]\right\|$ can be determined from the Lie algebra structure constants. The fact that $X_+$ is required to be in a neighbourhood of $\Omega_+$ provides a bound on $\left\|X_+\right\|$. The fact that all $X_+ \in \mathrm{hor}_{\Omega_+}V_2$ satisfy $\mathrm{YM1}_{\Omega_+}(X_+) = 0$, the fact that all $a\in \gf_0^{\Lambda_0}$ satisfy $\mathrm{pr}_{\mathrm{hor}_{\Omega_+}V_2}([a,\Omega_+])=0$, and the essentially polynomial nature of the various commutators then lead to the above propositions in a more or less routine way. 
\end{proof}
We can now write the equation \eqref{ym2gammajunk} as
\begin{align}
\dot{\Gamma}_+ &= \Gamma_+ + AZ_+ + J(z,\nu,Z_+,\Gamma_+),
\end{align}
where $J$ is the function
\begin{align*}
J(z,\nu,Z_+,\Gamma_+) :=  \left(\mathbb{I} + Q_{X_+}\right)^{-1}\!\left( Q_{X_+}\!\left(\Gamma_+\!\! +\! AZ_+\right) + R(Z_+) +\delta\Gamma_+ + J_{\eta}(X_+,\dot{X}_+)\!\right)\!.
\end{align*}

It follows from the above propositions that $J$ is a bounded and continuous function on a neighbourhood of the critical point,  $(z,\nu,Z_+,\Gamma_+)=(0,0,0,0)$, and  satisfies $J(0,0,0,0) = 0$ and $DJ(0,0,0,0) = 0$. We can then construct a function, $j$, which is bounded and continuous on all of $(0,\infty) \times (0,\infty) \times  \mathrm{hor}_{\Omega_+}V_2 \times \mathrm{hor}_{\Omega_+}V_2$, and agrees with $J$ on a neighbourhood of the critical point.
\par
Furthermore, since $J(0,0,0,0) = 0$ and $DJ(0,0,0,0) = 0$, we can ensure that $j$ is Lipschitz continuous with a fixed Lipschitz constant that is as small as necessary (by restricting the neighbourhood where $j = J$ as necessary). The nonlinear terms in the first two equations, $-z\nu$ and $\frac{-1}{2}\nu^2-z^2(\breve{G}+P)$, are also zero with vanishing linearization at the critical point, so we can replace them with functions $h_1, h_2$ that are bounded, continuous, Lipschitz with a Lipschitz constant that can be made arbitrarily small, and agree with the original nonlinear terms on some neighbourhood of the critical point.
\par
Then the nonlinear terms in the system, 
\begin{align}\label{goodsys1}
\dot{z} &= -z -z\nu,\\
\dot{\nu} &= -\nu -\frac{1}{2}\nu^2 -z^2(\breve{G}+P),\\
\dot{Z}_+ &= \Gamma_+,\\\label{goodsys4}
\dot{\Gamma}_+ &= \Gamma_+ + AZ_+ + J(z,\nu, Z_+,\Gamma_+),
\end{align}
i.e.,
\begin{align*}
\left(\begin{array}{c} -z\nu \\ \frac{-1}{2}\nu^2-z^2(\breve{G}+P) \\ 0 \\J(z,\nu, Z_+,\Gamma_+)\end{array}\right),
\end{align*}
will agree with 
\begin{align*}
h(z,\nu,Z_+,\Gamma_+) = \left(\begin{array}{c} h_1(z,\nu,Z_+,\Gamma_+) \\ h_2(z,\nu,Z_+,\Gamma_+) \\ 0 \\j(z,\nu, Z_+,\Gamma_+)\end{array}\right),
\end{align*}
in some neighbourhood of the critical point. Therefore we now consider the system given by
\begin{align}\label{lipschitzed}
\frac{d}{d\tau}\left(\begin{array}{c}z\\\nu\\Z_+\\\Gamma_+\end{array}\right) = \left[\begin{array}{cccc} -1 & 0 & 0 & 0\\ 0 & -1 & 0 & 0 \\ 0 & 0 & 0 & \mathbb{I}\\ 0 & 0 & D & \mathbb{I}\end{array}\right]\left(\begin{array}{c}z\\\nu\\Z_+\\\Gamma_+\end{array}\right) + h(z,\nu,Z_+,\Gamma_+),
\end{align}
where $Z_+$ and $\Gamma_+$ are assumed to be expanded over the $X$-type eigenbasis of $A$, which we know from \S{\ref{ym1eigsofA}} is a basis for $\mathrm{hor}_{\Omega_+}V_2$, and $D$ is the corresponding diagonal matrix of nonzero eigenvalues. We now have a dynamical system on $\mathbb{R}^n$, where $n = 1 + 1 + \mathrm{dim}_{\mathbb{R}}\left(\mathrm{hor}_{\Omega_+}V_2\right) + \mathrm{dim}_{\mathbb{R}}\left(\mathrm{hor}_{\Omega_+}V_2\right)$, to which we can apply Lemma \ref{dslemma3}. To see that the eigenvalues of the linearization all have nonzero real part, recall from the previous chapter that the eigenvalues associated to the $X$-type eigenvectors are given by
\begin{align}
\begin{array}{lc}
\textbf{eigenvector} & \textbf{eigenvalue}\\
\hline
X_{k-1, 2k}^a & -\frac{k+1}{2},\\
\breve{X}_{k-1, 2k}^a & -\frac{k+1}{2},\\
X_{2k,2k}^a & k(k+1),\\
X_{m,2k}^a & mk-k^2+\frac{m}{2},\\
\breve{X}_{m,2k}^a & mk-k^2+\frac{m}{2},\\
\end{array}
\end{align}
Each of the diagonal entries of $D$ will correspond to an eigenvalue in the table. Consider an eigenvalue from the table, $\lambda$. The corresponding contribution to the eigenvalues of the linearization will be the eigenvalues of the matrix
\begin{align}
\left[\begin{array}{cc} 0 & 1 \\ \lambda & 1 \end{array}\right],
\end{align}
i.e., 
\begin{align}
\frac{1\pm\sqrt{1+4\lambda}}{2}.
\end{align}
For the $X_{k-1, 2k}^a$ and $\breve{X}_{k-1, 2k}^a$ eigenvectors, the corresponding eigenvalues in the linearization are of the form
\begin{align}
\frac{1}{2} \pm \frac{\sqrt{1-2k}}{2},
\end{align}
which, since $k\geq 1$, all have real part equal to $\frac{1}{2}$. For the $X_{2k,2k}^a$ eigenvectors, the corresponding eigenvalues in the linearization are of the form
\begin{align}
-k, k+1,
\end{align}
which, since $k \geq 1$, are  pairs of (integer) eigenvalues, one less than or equal to $-1$ and the other greater than or equal to $2$. For the $X_{m,2k}^a$ and $\breve{X}_{m,2k}^a$ eigenvectors, the corresponding eigenvalues in the linearization are of the form
\begin{align}
\frac{1}{2} \pm \frac{\sqrt{1+4mk -4k^2+2m}}{2},
\end{align}
which, since $k \geq 1$ and $m > k$, are pairs of eigenvalues, one greater than or equal to $2$, and the other less than or equal to $-1$. Note that the solutions will often not be integers.
The contributions from the $z$ and $\nu$ equations are clearly both eigenvalues equal to $-1$.
We finally have a system which satisfies the preconditions of \ref{dslemma3}, resulting in the following theorem.
\begin{theorem}
The bounded solutions of the system \eqref{lipschitzed} on $[0,\infty)$ are completely determined by the value of\,\, $\mathbb{P}(z,\nu,Z_+,\Gamma_+)$ at any point $q \in [0,\infty),$ where $\mathbb{P}$ is the projection onto the negative eigenspaces of the linearization.
\end{theorem}

\begin{theorem}\label{dynamexist}
Every bounded solution to the original (projected) dynamical system, \eqref{eq753} - \eqref{ym2gammajunk}, is determined uniquely by $c$ constants, where $c$ is the number of negative eigenvalues in the linearization.
\end{theorem}
\begin{proof}
The bounded solutions of the dynamical system in the previous theorem, \eqref{lipschitzed}, are completely determined by $\mathbb{P}(z,\nu,Z_+,\Gamma_+)$ at any point $q \in [0,\infty)$. Each of the solutions of \eqref{eq753} - \eqref{ym2gammajunk} eventually agrees with a solution of \eqref{lipschitzed} (via \eqref{goodsys1} - \eqref{goodsys4}). By choosing the point $q$ so that the solutions agree at $\tau =q$ the result follows. 
\end{proof}
\section{Asymptotic properties of the solutions}
In this section we will determine the asymptotic properties of the solutions. The standard proof of Lemma \ref{dslemma3} involves the contraction mapping principle, where the iterative map is obtained from solving the linear inhomogeneous equation
\begin{align}
y_{k+1}' = By_{k+1} + h(t,y_{k}(t)),
\end{align}
which has the integral formula
\begin{align*}\label{Tymap}
y_{k+1} = T y_{k} &= e^{(t-q)B}y_-(q) + \int_q^{\tau} e^{(t-s)B}\mathbb{P}h(s,y_{k}(s))ds -\int_{\tau}^{\infty}(\mathbb{I}-\mathbb{P})h(s,y_k(s)) ds,
\end{align*}
and $y = \lim_{k\rightarrow\infty}y_k$.
By choosing the initial iterate to be a solution of the linearized equations we can track the asymptotic fall-off of the solutions over a finite number of iterations.
\par
The linearized system is
\begin{align}
\frac{d}{d\tau}\left(\begin{array}{c}z\\\nu\\Z_+\\\Gamma_+\end{array}\right) = \left[\begin{array}{cccc} -1 & 0 & 0 & 0\\ 0 & -1 & 0 & 0 \\ 0 & 0 & 0 & \mathbb{I}\\ 0 & 0 & D & \mathbb{I}\end{array}\right]\left(\begin{array}{c}z\\\nu\\Z_+\\\Gamma_+\end{array}\right),
\end{align}
which, as discussed in the previous section, has negative eigenvalues, $-1$, $-1$, $-k$, $\frac{1}{2} - \frac{\sqrt{1+4mk -4k^2+2m}}{2}$, which are all less than or equal to $-1$. To discuss the asymptotic behaviour we define the following terminology:
\begin{definition}
By saying $x$ is $O(e^{-c\tau})$, or that `$x$ has $O(e^{-c\tau})$ fall-off', we mean $\lim_{\tau\rightarrow\infty}x e^{c\tau}$ exists.
\end{definition}
Then, from the above enumeration of possible eigenvalues for the linearization, we have:
\begin{proposition}
The solutions to the linearized equation are at least $O(e^{-\tau})$.
\end{proposition}
\begin{proposition}
If the fall-off for each of $z,\nu,Z_+,\Gamma_+$ is at least $O(e^{-c\tau})$ then the fall-off for $h(z,\nu,Z_+,\Gamma_+)$ is also at least $O(e^{-c\tau})$.
\end{proposition}
\begin{proof}
In the neighbourhood of the critical point where the Lipschitz functions $j, h_1, h_2$ match the original functions used to define them, we have from \eqref{goodsys1}-\eqref{goodsys4},
\begin{align}
h(z,\nu,Z_+,\Gamma_+) = \left(\begin{array}{c} -z\nu \\ -\frac{1}{2}\nu^2 -z^2(\breve{G}+P) \\ 0 \\ J(z,\nu,Z_+,\Gamma_+)\end{array}\right).
\end{align}
Since the solutions for $z$ and $\nu$ in the linearized equations are both at least $O(e^{-c\tau})$, $z^2, z\nu$, and $\nu^2$ are all at least $O(e^{-2c\tau})$. Hence the first component of $h$ is at least $O(e^{-2c\tau})$. For the additional terms in the second component of $h$, we have
\begin{align}
P = \left\|\Lambda_0 - \Omega_0 - [Z_+,\Omega_-] - [\Omega_+,Z_-] -[Z_+,Z_-]\right\|^2,
\end{align}
which is $O(1)$ due to the constant term. The fall-off of $\breve{G}$ depends on $[\eta,X_+]$, since,
%\
\begin{align}
\breve{G} = \frac{1}{2}\left\|\Gamma_+\right\|^2 - \frac{1}{2}\left\|[\eta,X_+]\right\|^2.
\end{align}
We will determine the fall-off for this term once we have established the fall-off for $[\eta,X_+]$
\par
The third component of $h$ is identically zero. 
\par
The fourth term of $h$, $j$, (in the region where $j=J$), is
\begin{align*}
J(z,\nu,Z_+,\Gamma_+) =  \left(\mathbb{I} + Q_{X_+}\right)^{-1}\!\left( Q_{X_+}\!\left(\Gamma_+\!\! +\! AZ_+\right) + R(Z_+) +\delta\Gamma_+ + J_{\eta}(X_+,\dot{X}_+)\!\right)\!.
\end{align*}
The terms are all acted on by $\left(\mathbb{I} + Q_{X_+}\right)^{-1}$, which, since we are in a neighbourhood where $\left\|Q_{X_+}\right\| < 1$, is expanded as $\mathbb{I} - Q_{X_+} + (Q_{X_+})^2 - (Q_{X_+})^3 + \ldots$. Consequently we would like to show that the operator $Q_{X_+}$ preserves the fall-off, i.e. if $W_+ \in \mathrm{hor}_{\Omega_+}V_2$ is $O(e^{-c\tau})$, then $Q_{X_+}W_+$ is also at least $O(e^{-c\tau})$. This is certainly plausible, since $Q_{X_+}$ satisfies the identity
\begin{align}
\left(\mathrm{YM1}_{X_+}\circ Q_{X_+}\right)W_+ = \mathrm{YM1}_{X_+}W_+,
\end{align}
but we must take some care to ensure that $\mathrm{YM1}_{X_+}$ does not project out any terms in $Q_{X_+}W_+$ with bad fall-off. To do this we use the following the Cauchy-Schwarz type trick -- firstly, by definition of $\tilde{\mathcal{S}}^{-1}_{X_+}$ we have the identity
\begin{align}
[[\tilde{\mathcal{S}}^{-1}_{X_+}\!\circ \mathrm{YM1}_{X_+}(W_+),X_+,X_-] + [[\tilde{\mathcal{S}}^{-1}_{X_+}\circ \mathrm{YM1}_{X_+}(W_+),X_-],X_+] = \mathrm{YM1}_{X_+}(W_+).
\end{align}
The properties of the inner product imply that if $\mathrm{YM1}_{X_+}(W_+)$ is in the kernel of $\mathrm{ad}_{X_+}$ then $\mathrm{YM1}_{X_+}(W_+)=0$. By taking the inner product of both sides with $\tilde{\mathcal{S}}^{-1}_{X_+}\circ \mathrm{YM1}_{X_+}(W_+)$ and simplifying, we get the identity
\begin{align}
\left\|[\tilde{\mathcal{S}}^{-1}_{X_+}\circ \mathrm{YM1}_{X_+}(W_+),X_+]\right\|^2 = \brac [\tilde{\mathcal{S}}^{-1}_{X_+}\circ \mathrm{YM1}_{X_+}(W_+),X_+] | W_+ \ket.
\end{align}  
Then, if $W_+$ is $O(e^{-c\tau})$, we can multiply both sides of the equation by $e^{c\tau}$ and take the limit as $\tau\rightarrow\infty$.
$$\lim_{\tau\rightarrow\infty}e^{c\tau}\!\left\|[\tilde{\mathcal{S}}^{-1}_{X_+}\circ \mathrm{YM1}_{X_+}(W_+),X_+]\right\|^2\!= \lim_{\tau\rightarrow\infty}e^{c\tau}\!\brac [\tilde{\mathcal{S}}^{-1}_{X_+}\circ \mathrm{YM1}_{X_+}(W_+),X_+] | W_+ \ket.$$
Since we know the fall-off of $W_+$, both sides have a well defined limit, hence the left hand side is $O(e^{-c\tau})$, and hence $[\tilde{\mathcal{S}}^{-1}_{X_+}\circ \mathrm{YM1}_{X_+}(W_+\!),X_+]$ is at least $O(e^{-\frac{c}{2}\tau})$. Since $[\tilde{\mathcal{S}}^{-1}_{X_+}\circ \mathrm{YM1}_{X_+}(W_+\!),X_+]$ appears on the right hand side as well, we can then multiply both sides of the equation by an additional factor of  $e^{\frac{c}{2}\tau}$ and again get a well-defined limit, hence $[\tilde{\mathcal{S}}^{-1}_{X_+}\circ \mathrm{YM1}_{X_+}(W_+),X_+]$ is at least $O(e^{-\frac{3c}{4}\tau})$.
\par
Repeating this process over and over again we will eventually determine that $[\tilde{\mathcal{S}}^{-1}_{X_+}\circ \mathrm{YM1}_{X_+}(W_+),X_+]$ is at least $O(e^{-(c-\epsilon)\tau})$ for arbitrarily small $\epsilon$. By the definition of $Q_{X_+}$,
\begin{align*}
Q_{X_+}:=\left(\mathrm{pr}_{\mathrm{hor}_{\Omega_+}V_2}\circ \mathrm{ad}_{X_+}\circ\tilde{\mathcal{S}}^{-1}_{X_+}\circ\mathrm{YM1}_{X_+}\right),
\end{align*}
it then follows that $Q_{X_+}W_+$ is at least $O(e^{-(c-\epsilon)\tau})$. The nature of the compositions in $Q_{X_+}$ ensures that we can take this process to the natural conclusion and say $Q_{X_+}W_+$ is at least $O(e^{-c\tau})$.
\par
The Cauchy-Schwarz type trick also works for determining the fall-off of $\eta$, we use the first order Yang-Mills equation in the reduced variables and the properties of the inner product to obtain
\begin{align}
\left\|[\eta,X_+]\right\|^2 = -\brac [\eta, X_+] | \Gamma_+ \ket,
\end{align}
and then, since $\Gamma_+$ is $O(e^{-c\tau})$, the same procedure shows that $[\eta,X_+]$ (and then consequently $\breve{G}$) is at least $O(e^{-c\tau})$ By \eqref{etaasxx} and \ref{tenprops}, it follows that $\eta$ is at least $O(1)$.
\par
Finally, $R$ is at least $O(e^{-2c\tau})$, $\delta$ is at least $O(e^{-2c\tau})$. Putting this all together, we obtain the result for $h$.
\end{proof}
So the iteration procedure can only improve or maintain the fall-off. The smallest negative eigenvalue for the linearization is $-1$, so after any finite number of iterations we have that $y_k$ is at least $O(e^{-\tau})$. We cannot immediately make the same conclusion for the fixed point $y = \lim_{k\rightarrow\infty}y_k$ since the fall-off is also defined in terms of a limit and we have not established that the $k$ limit and the $\tau$ limit in 
\begin{align}
\lim_{\tau\rightarrow\infty}e^{\tau}\lim_{k\rightarrow \infty}y_k
\end{align}
can be exchanged. By sacrificing an arbitrarily small amount of fall-off we can obtain a sufficiently good result as follows.
\begin{definition}
Fix $ 0 < s < 1$ and define the Banach space $\mathfrak{B}_s$ of continuous functions from $[0,\infty) \rightarrow \R^n$ equipped with the weighted supremum norm
\begin{align}
\left\|y\right\|_s := \sup_{[0,\infty)} |e^{s\tau}y|.
\end{align}
\end{definition}
\begin{proposition}
The map \eqref{Tymap} is also a contraction on $\mathfrak{B}_s$.
\end{proposition}
\begin{proof}
Let $y = (z,\nu, Z_+,\Gamma_+)$ and write the system \eqref{lipschitzed} as
\begin{align}
\dot{y} = By + h(y).
\end{align}
Introduce the variable $\tcal{y} := e^{s\tau}y$. It follows that
\begin{align}
\dot{\tcal{y}\,\,} = \left(B+s\mathbb{I}\right)\tcal{y} + e^{s\tau}h(y).
\end{align}
The final term, $e^{s\tau}h(y)$, can be rewritten in terms of $\tcal{y}$. For example, in the neighbourhood where \eqref{lipschitzed} is equal to the original system, \eqref{goodsys1}-\eqref{goodsys4}, the first component is 
\begin{align*}
e^{s\tau}h_1(z,\nu, Z_+, \Gamma_+) &= e^{s\tau} (-z\nu)\\
&= -e^{-s\tau} (e^{s\tau}z)(e^{s\tau}\nu)\\
&=e^{-s\tau}h_1(\,\tcal{y}).
\end{align*}
In general the components of $h(y)$ will split into pieces that absorb the $e^{s\tau}$ inhomogeneously. The overall transformation is
\begin{align}
e^{s\tau}h(y) = \tilde{h}(\tau,\tcal{y})
\end{align}
Since $\mathrm{YM1}_{e^{-s\tau}Z_+} = e^{-s\tau}\mathrm{YM1}_{Z_+}$ and $\mathrm{ad}_{e^{-s\tau}Z_+} = e^{-s\tau}\mathrm{ad}_{Z_+}$, we can again use the properties \emph{a. -- j.} in Proposition. \ref{tenprops} and obtain the result that, in a neighbourhood of the critical point, $\tilde{h}(\tau,\tcal{y})$ is a bounded and continuous function of $\tau$ and $\tcal{y}\,$ that vanishes and has vanishing linearization at the critical point for all $\tau$. Replace $\tilde{h}$ with a Lipschitz function $\tcal{h}\,\,$, that agrees with $\tilde{h}$ on a neighbourhood of the critical point and has Lipschitz constant $l = 1-s$.

The system in the $\tcal{y}$ variable,
\begin{align}
\dot{\tcal{y}\,\,} = \left(B+s\mathbb{I}\right)\tcal{y} + \tcal{h}\,\,(\tau, \tcal{y}),
\end{align}
has linearization given by the matrix $B+s\mathbb{I}$. Since any vector is an eigenvector of the identity matrix, the eigenvectors of $B+s\mathbb{I}$ are the eigenvectors of $B$. Since $0 < s < 1$, and $B$ only has negative eigenvalues that are less than or equal to $-1$ and positive eigenvalues that have real part greater than or equal to $\frac{1}{2}$, the eigenvalues of $B+s\mathbb{I}$ are negative eigenvalues less than zero, and positive eigenvalues greater than $\frac{1}{2}$.
We have already used Lemma \ref{dslemma3} to solve this equation in the neighbourhood where this system is equal to \eqref{lipschitzed} in different variables, but by repeating the steps of the proof for the system in the new variables, we get that the fixed point of the contraction mapping, $$\tcal{y} = \lim_{k\rightarrow\infty}\,\,\tcal{y}_k \in \mathfrak{B}_s.$$
\end{proof}
By choosing $s$ arbitrarily close to $1$ we now have fall-off for $y$, i.e the limit $$\lim_{\tau\rightarrow\infty} e^{s\tau}(z,\nu,Z_+,\Gamma_+)$$
exists and hence
$$\lim_{\tau\rightarrow\infty} e^{(1-\epsilon)\tau}(z,\nu,Z_+,\Gamma_+) =0$$
for arbitrarily small positive $\epsilon$. This then implies the following theorem 
\begin{theorem}
For any bounded solution to the static, spherically symmetric Einstein Yang-Mills equations in the reduced variables, for Abelian models, or non-Abelian models arising from classical groups, the fall-off as $r$ tends to infinity is at least $O(e^{-(1-\epsilon)\tau})$. Moreover, in the reduced variables we have
\begin{align}
z &= O(e^{-(1-\epsilon)\tau})\\
\nu &= O(e^{-(1-\epsilon)\tau})\\
X_+ &= \Omega_+ + O(e^{-(1-\epsilon)\tau})\\
\Gamma_+ &= O(e^{-(1-\epsilon)\tau})
\end{align}
\end{theorem}
This theorem then implies:
\begin{theorem}
Any bounded solution to the static, spherically symmetric Einstein Yang-Mills equations for models arising from classical gauge groups, will have a well-defined limit as $r$ tends to infinity. 
\end{theorem}
\begin{proof}
For any bounded solution to the equations we now know that $\Gamma_+$ is at least $O(e^{-(1-\epsilon)\tau})$, and by the procedure above, so is $[\eta, X_+]$. The square of the speed of the curve $\Lambda(\tau)$ can be written in terms of the reduced variables as 
\begin{align}
\left\|\dot{\Lambda}_+\right\|^2 = \left\|\Gamma_+\right\|^2 - \left\|[\eta,X_+]\right\|^2
\end{align}
which then implies that $\dot{\Lambda}_+$ is at least $O(e^{-(1-\epsilon)\tau})$. This implies that the arc-length,
\begin{align}
\int_{q}^{\infty} \left\|\dot{\Lambda}(\tau)_+\right\| d\tau,
\end{align}
is a finite quantity. This implies that bounded solutions have a well-defined limit as $\tau$ goes to infinity. To convert back to the $r$ variable, we use the equations
\begin{align*}
\frac{dr}{d\tau} &= r(1+\nu)\\
&= r(1+O(e^{-(1-\epsilon)\tau})),
\end{align*}
and
\begin{align*}
r\Lambda_+' &= \dot{\Lambda}_+(1+\nu)^{-1}\\
&= \dot{\Lambda}_+(1+O(e^{-(1-\epsilon)\tau})).
\end{align*}
The first of these implies that $$r = e^{\tau} + \text{terms with better fall-off,}$$
which then implies $$r^{1-2\epsilon} = e^{(1-2\epsilon)\tau} + \text{terms with better fall-off,}$$
We can then multiply the second equation by $r^{1-2\epsilon}$ and use the fact that  $\dot{\Lambda}_+$ is at least $O(e^{-(1-\epsilon)\tau})$, to improve the estimate \eqref{OKasymptote} of $r\Lambda_+' \rightarrow 0$ to 
$$r^{2-2\epsilon}\Lambda_+' \rightarrow 0,$$
for arbitrarily small positive $\epsilon$, which is more than enough fall-off to establish that all bounded solutions not only go to one of the $G_0^{\Lambda_0}$-orbits through one of the $\Omega_+^i$ but that they limit to one specific point on the orbit, (which under a choice of constant gauge can be assumed to be the point $\Omega_+^i$).
\end{proof}
\section{Particle-like Solutions in Abelian Models always have Zero Magnetic Charge}
In this section we complete the necessary results to prove Theorem \ref{nomagabelian}. 
\begin{proposition}\label{nostrangeeigenvalues} There are no $mk-k^2+\frac{m}{2}$ type eigenvalues ( associated to the $X_{m,2k}$-type eigenvectors with $m>2k$) of the linear operator $A$ for Abelian models.
\end{proposition}
\begin{proof} The proof is by contradiction. If there were such an eigenvalue, then we could use $\mathrm{ad}_{\Omega_-}^k$ to lower the corresponding highest weight $\mu_{m,2k}^a$ down to the zero eigenspace of $\Lambda_0$. But since the residual group is Abelian, $\mathrm{ad}_{\Omega_-}^k \cdot \mu_{m,2k}^a$ would have to be in the Cartan subalgebra, implying a zero weight for $\mathrm{ad}_{\Omega_0}$ and therefore implying that $m=2k$ which is a different type of eigenvalue to the assumption.
\end{proof}
This means that apart from the two $-1$ eigenvalues from the $z$ and $\nu$ equations, the only other negative eigenvalues of the linearization of the dynamical system are the negative integers, $-k$, coming from the $k(k+1)$ eigenvalues of $A$.
In \cite{Oli02a}, Oliynyk and K\"unzle proved the existence and uniqueness of analytic solutions for small $z=r^{-1}$ for the zero magnetic charge case of an Abelian model. In \cite{MFthesis} we showed that for an Abelian model a trivial modification of this theorem could be made so that instead of establishing the result on $V_2^{\Omega_0}$ we establish it on $V_2^{\Lambda_0} \cap V_2^{\Omega_0}$. While it is a zero magnetic charge asymptotic solution for the equations on $V_2^{\Omega_0}$ it is a nonzero magnetic charge asymptotic solution for the equations on  $V_2^{\Lambda_0}$. 
\par
The number of parameters for the otherwise unique analytic asymptotic solutions and the solutions we have proved the existence of in Theorem \ref{dynamexist} match exactly and so the existence results are equivalent for the Abelian case. 
\par
The step in the original theorem in \cite{Oli02a} where the power series solutions are shown to be well-posed then becomes (in the modified theorem) the statement that that the power series remains in the intersection $V_2^{\Lambda_0} \cap V_2^{\Omega_0}$. All of this amounts to the following result:
\begin{proposition}
In the case of an Abelian model, there exists some neighbourhood $[q,\infty)$, such that the asymptotic solutions of \eqref{meqn}-\eqref{ym1} satisfying $\lim_{r\rightarrow\infty}\Lambda_+=\Omega_+^i$ are contained within $V_2^{\Lambda_0} \cap V_2^{\Omega_0^i}$.
\end{proposition}
We now show that the equations \eqref{meqn}-\eqref{ym1} preserve this condition:
\begin{proposition}
If on some nonzero interval, say $(r_1,r_2)$, a solution of \eqref{meqn}-\eqref{ym1} is contained within $V_2^{\Lambda_0} \cap V_2^{\Omega_0^i}$, then the solution will be contained within $V_2^{\Lambda_0} \cap V_2^{\Omega_0^i}$ for the entire interval of existence.
\end{proposition}
\begin{proof}
The only thing to show is that the $\mathcal{F}$ term in $\eqref{ym2}$, i.e.,
\begin{align}
\mathcal{F}: V_2 \rightarrow V_2,\quad \mathcal{F}(\Lambda_+) = \Lambda_+ - \frac{1}{2}[[\Lambda_+,\Lambda_-],\Lambda_+],
\end{align}
will preserve the intersection. By using the Jacobi Identity to evaluate 
$$[\Lambda_0,\mathcal{F}(V_2^{\Lambda_0}\cap V_2^{\Omega_0})], [\Omega_0,\mathcal{F}(V_2^{\Lambda_0}\cap V_2^{\Omega_0})]$$ it is clear that 
$V(2,2) = V_2^{\Lambda_0}\cap V_2^{\Omega_0}$ is mapped to itself. It follows that $V_2^{\Lambda_0}\cap V_2^{\Omega_0}$ is a preserved subspace for the differential equation \eqref{ym2}.
\end{proof}
If a local solution to the system \eqref{meqn}-\eqref{ym2} is indeed a globally regular solution, than $\Lambda_+$ must remain bounded all the way back to $r=0$, where the condition \eqref{requalszerocondition} will need to be satisfied. We will now prove that it is impossible to solve this equation on $V_2^{\Lambda_0} \cap V_2^{\Omega_0}$ for distinct $\Lambda_0$ and $\Omega_0$.
\begin{theorem} There are no solutions, $X_+\in V_2^{\Lambda_0} \cap V_2^{\Omega_0}$, (with $X_-=-c(X_+)$) of the equation 
\begin{align}
[X_+, X_-] = \Lambda_0
\end{align}
unless $\Omega_0 = \Lambda_0$.
\end{theorem}
\begin{proof}
Suppose there exists an $X_+\in V_2^{\Lambda_0} \cap V_2^{\Omega_0}$ such that
\begin{align}
[X_+, X_-] = \Lambda_0.
\end{align} 
Since $X_+$ is in the intersection $V_2^{\Lambda_0} \cap V_2^{\Omega_0}$ it is true that $X_+\in V_2^{\Lambda_0}$ and $X_+\in V_2^{\Omega_0}$. These inclusions translate to the equations
\begin{align}
[\Lambda_0, X_{\pm}] = \pm 2 X_{\pm}
\end{align}
and 
\begin{align}
[\Omega_0, X_{\pm}] = \pm 2 X_{\pm}.
\end{align}
By the definition of $V_2^{\Omega_0}$ we also have some element $\Omega_+\in V_2^{\Lambda_0} \cap V_2^{\Omega_0}$ satisfying

\begin{align}
[\Omega_+, \Omega_-] = \Omega_0,\\
[\Omega_0, \Omega_{\pm}] = \pm 2 \Omega_{\pm}.
\end{align}
We want to show that the only way these equations can all hold is if $\Omega_0 = \Lambda_0$. By the properties of the inner product we have
\begin{align*}
\left\|\Lambda_0\right\|^2 &= \brac \Lambda_0 | \Lambda_0 \ket\\
& = \brac \Lambda_0 | [X_+, X_-] \ket\\\yesnumber\label{resultAAA}
& = 2 \left\|X_+\right\|^2,
\end{align*}
and also
\begin{align*}
\left\|\Lambda_0\right\|^2 &= \brac \Lambda_0 | \Lambda_0 \ket\\
& = \brac [X_+,X_-] | [X_+, X_-] \ket\\\yesnumber\label{resultBBB}
& = \left\|[X_+,X_-]\right\|^2.
\end{align*}
Since $X_+$ is in both $V_2$ spaces it also satisfies both ``coercive conditions'' (recall from \S \ref{innerproductproperties}):
\begin{align}\label{coerceCCC}
\left\|[X_+,X_-]\right\|^2 &\geq \frac{4\left\|X_+\right\|^4}{\left\|\Lambda_0\right\|^2},\\\label{coerceDDD}
\left\|[X_+,X_-]\right\|^2 &\geq \frac{4\left\|X_+\right\|^4}{\left\|\Omega_0\right\|^2},
\end{align}
Combining \eqref{resultAAA}, \eqref{resultBBB}, and \eqref{coerceDDD} we arrive at 
\begin{align}\label{resultEEE}
\left\|\Lambda_0\right\|^2 \leq \left\|\Omega_0\right\|^2
\end{align}
Next we consider the projection of $\Omega_0$ onto $\Lambda_0$ via the $\Omega_{\pm}$ quantities:
\begin{align*}
\brac \Omega_0 | \Lambda_0 \ket
=& \brac [\Omega_+, \Omega_-] | \Lambda_0 \ket\\
=& \brac \Omega_+, | [\Lambda_0,\Omega_+] \ket\\
=& \brac \Omega_+, | 2 \Omega_+ \ket\\
=& \brac \Omega_+, | [\Omega_0, \Omega_+] \ket\\
=& \brac [\Omega_+, \Omega_-] | \Omega_0 \ket\\
=& \brac \Omega_0 | \Omega_0 \ket\\
=& \left\|\Omega_0\right\|^2,
\end{align*}
which implies that
\begin{align}\label{equation6126}
\Omega_0 = \frac{\left\|\Omega_0\right\|^2}{\left\|\Lambda_0\right\|^2}\Lambda_0 + \text{ orthogonal terms}.
\end{align}
So, we have
\begin{align}
\left\|\Omega_0\right\|^2 \leq \left\|\Lambda_0\right\|^2,
\end{align}
which when combined with \eqref{resultEEE} implies that $\left\|\Lambda_0\right\| = \left\|\Omega_0\right\|$, which can be substituted into the expansion \eqref{equation6126} to give
\begin{align}\nonumber
\Lambda_0 = \Omega_0 + \text{ orthogonal terms. }
\end{align}
Hence $\Lambda_0 = \Omega_0$.
\end{proof}
Combining the above results then proves Theorem \ref{nomagabelian} -- it is impossible for a globally regular solution of \eqref{meqn}-\eqref{ym1} to satisfy both the the condition of nonzero magnetic charge and the boundary condition at $r=0$ for an Abelian model.
\section{The possibility of magnetically charged particle-like solutions for non-Abelian Models}
The local existence near infinity of solutions with nonzero total magnetic charge, guaranteed by Theorem \ref{dynamexist}, establishes the possibility of magnetically charged, globally regular solutions of the static, spherically symmetric Einstein Yang-Mills equations. While we have ruled out such solutions for Abelian models, the possibility remains open for the non-Abelian case. We will now briefly mention some results regarding prospective models.
\par
The solutions will not typically be asymptotically analytic in $\frac{1}{r}$. Combining the requirement that asymptotic solutions are analytic in $\frac{1}{r}$ but are not contained in some $V_2^{\Lambda_0} \cap V_2^{\Omega_0^i}$ leads to a number-theoretic condition on the eigenvalues of $A$ related to a subfamily of the Pythagorean triples. The smallest example we could find of a model satisfying this condition is a non-Abelian model in $SU(30)$ with $V_2 \cong \R^{112}$! Even in the reduced variables this would be a system of 54 nonlinear equations.
\par
The case where asymptotic analyticity in $\frac{1}{r}$ is not imposed is more promising. We have candidate models in the rank $2$ exceptional Lie group $G_2$ and another in $SU(6)$. In these cases there is at least one component of the Yang-Mills fields which asymptotically tends to an irrational power of $1/r$, e.g., $r^{\frac{1}{2} - \frac{\sqrt{3}\sqrt{5}}{2}}$.
\par
Whether or not these asymptotic solutions give rise to globally regular solutions remains to be seen.
\section{Acknowledgements}
This work was partially supported by the ARC grant DP1094582 and an MRA grant.
\bibliography{references}
\bibliographystyle{amsplain}
\end{document}